\documentclass[12pt,onecolumn,draftcls]{IEEEtran}

\usepackage{hyperref}
\usepackage{subfig}
\usepackage{graphicx}
\usepackage{amsthm,amsfonts,amssymb,amsmath}
\usepackage{algorithm}
\usepackage{algpseudocode}
\usepackage{color}
\usepackage{bm,bbm}
\newtheorem{proposition}{Proposition}

\newtheorem{theorem}{Theorem}
\newtheorem{remark}{Remark}
\newtheorem{lemma}{Lemma}
\newtheorem{definition}{Definition}
\newtheorem{example}{Example}
\newtheorem{corollary}{Corollary}
\newcommand{\Z}{\mathbb{Z}}
\newcommand{\N}{\mathbb{N}}
\newcommand{\Q}{\mathbb{Q}}
\newcommand{\QQ}{Q_{(n,e)}^{k}}
\newcommand{\QB}{Q_{B}^{k}}
\newcommand{\R}{\mathbb{R}}

\newcommand{\lp}[1]{\mbox{LPL}(n,#1)}

\newcommand{\simb}[2]{\left< _{#2}^{#1}\right>}
\newcommand{\doblesum}[2]{\sum_{\substack{#1 \\ #2 }}}

\begin{document}

\title{On the non-existence of linear perfect Lee codes: The Zhang-Ge condition and a new polynomial criterion}
\author{Claudio~Qureshi}
\author{
    \IEEEauthorblockN{Claudio~Qureshi\IEEEauthorrefmark{1}}\\
    {\small
    \IEEEauthorblockA{\IEEEauthorrefmark{1}Institute of Mathematics, Statistics and Scientific Computing, University of Campinas, Brazil. 
    \\cqureshi@ime.unicamp.br} \\
}}

\maketitle

\vspace{-1.4cm}
\begin{abstract}
The Golomb-Welch conjecture (1968) states that there are no $e$-perfect Lee codes in $\Z^n$ for $n\geq 3$ and $e\geq 2$. This conjecture remains open even for linear codes. A recent result of Zhang and Ge establishes the non-existence of linear $e$-perfect Lee codes in $\Z^n$ for infinitely many dimensions $n$, for $e=3$ and $4$. In this paper we extend this result in two ways. First, using the non-existence criterion of Zhang and Ge together with a generalized version of Lucas' theorem we extend the above result for almost all $e$ (i.e. a subset of positive integers with density $1$). Namely, if $e$ contains a digit $1$ in its base-$3$ representation which is not in the unit place (e.g. $e=3,4$) there are no linear $e$-perfect Lee codes in $\Z^n$ for infinitely many dimensions $n$. Next, based on a family of polynomials (the $Q$-polynomials), we present a new criterion for the non-existence of certain lattice tilings. This criterion depends on a prime $p$ and a tile $B$. For $p=3$ and $B$ being a Lee ball we recover the criterion of Zhang and Ge.

\end{abstract}

\section{Introduction and Preliminaries}\label{SectionIntro}

Let $\Z$ and $\Z_q$ denote the ring of integer numbers and integers modulo $q$, respectively. For any two words $x=(x_1,\ldots,x_n)$ and $y=(y_1,\ldots,y_n)$, the Lee metric (also known as Manhattan or $\ell_1$ metric) is given by $$d(x,y)=\left\{ \begin{array}{ll}
\sum_{i=1}^n \min\left\{|x_i-y_i|, q-|x_i-y_i| \right\} & \textrm{for }x,y \in \Z_q^n, \\
\sum_{i=1}^n |x_i-y_i| & \textrm{for }x,y \in \Z^n.
\end{array} \right.$$

An $e$-perfect Lee code is a subset $C\subseteq \Z_q^n$ (or $C\subseteq \Z^n$) such that for each $x\in \Z_q^n$ (or $x\in \Z^n$) there is a unique $c=c(x)\in C$ satisfying $d(x,c)\leq e$. If in addition $C$ is an additive subgroup of $\Z_q^n$ (or $\Z^n$) we say that $C$ is a linear $e$-perfect Lee code. When $q\geq 2e+1$, the natural projection $\pi: \Z^n \rightarrow \Z_q^n$ (taking modulo $q$ in each coordinate) establishes a correspondence between $e$-perfect Lee codes in $\Z^n$ and $e$-perfect Lee codes in $\Z_{q}^n$. This correspondence preserves linearity. In this paper we denote by $B^n(e)$ the $n$-dimensional Lee ball of radius $e$ centered at the origin, that is, $B^n(e)=\{x\in \Z^n : d(x,0)\leq e\}$; and set $k(n,e) = \# B^n(e)$. The set of all \emph{linear} $e$-perfect Lee codes in $\Z^n$ is denoted by $\lp{e}$.\\

The Lee metric was introduced for transmission of signals over noisy channels in \cite{Lee58} for codes with alphabet $\Z_p$ with $p$ a prime number, then it was extended to alphabets $\Z_q$ ($q\in \Z^+$) and $\Z$ in \cite{GW68}, \cite{GW70}. One of the most central question on codes in the Lee metric is regarding the existence of such codes. In \cite{GW70}, Golomb and Welch showed that there are $e$-perfect Lee codes $C\subseteq \Z^2$ for each value of $e\geq 1$, and there are $1$-perfect Lee codes $C\subseteq \Z^n$ for each value of $n\geq 2$. They also proved that for fixed dimension $n\geq 3$, there is a radius $e_n>0$ ($e_n$ unspecified) such that there are no $e$-perfect Lee codes $C\subseteq \Z^n$ for $e\geq e_n$ and conjectured that it is possible to take $e_n=2$ (i.e. no $e$-perfect Lee codes in $\Z^n$ exist for $n\geq 3$ and $e\geq 2$). This conjecture has been the main motive power behind the research in the area. A recent survey of papers on the Golomb-Welch conjecture is provided in \cite{HK18}. Next we mention some of them. Explicit bounds for $e_n$ for periodic perfect Lee codes were obtained by  K. A. Post \cite{Post75}, namely $e_n=n-1$ for $3\leq n \leq 5$ and $e_n= \frac{\sqrt{2}}{2} n - \frac{1}{4}(3\sqrt{2}-2)$ for $n\geq 6$; and by P. Lepist{\"o} \cite{Lepisto81} who proved that an $e$-perfect Lee code must satisfy $n\geq (e+2)^2 / 2.1$ if $e\geq 285$. Recently, P. Horak and D. Kim \cite{HK18} proved that the above results hold in general, that is, without the restriction of periodicity. The Golomb-Welch conjecture was also proved for dimensions $3\leq n \leq 5$ and radii $e\geq 2$ \cite{GMP98,Spacapan07,Horak09a} and for $(n,e)=(6,2)$ \cite{Horak09b}. Recently, several papers have focus on the study of the Golomb-Welch conjecture for fixed radius $e$ and large dimensions $n$. The case $e=2$ has been treated in \cite{HG14, Kim17} and \cite{QCC18}. In the linear case, it is proved that $\lp{2}=\emptyset$ for infinitely many dimensions $n$. A new criterion for the non-existence of perfect Lee codes was presented by T. Zhang and G. Ge in \cite{ZG17} and it was used for the authors to obtain non-existence results for radii $e=3$ and $e=4$. This criterion states that if a pair $(n,e)$ of positive integers verifies the congruence system (\ref{EqZhangGesystem}) and $k(n,e)$ is squarefree, then $\lp{e}=\emptyset$ (Theorem \ref{ZhanGeTheorem}). In this paper, we say that a pair $(n,e)$ satisfies the \emph{Zhang-Ge condition} when it is a solution of the system (\ref{EqZhangGesystem}). Then, we define the Zhang-Ge set associated with $e$, denoted by $\mbox{ZG}(e)$, as the set of positive integers $n$ such that $(n,e)$ satisfies the Zhang-Ge condition. Determining the cardinality of the Zhang-Ge set $\mbox{ZG}(e)$ brings us information about $\lp{e}$.\\

This paper is organized as follows. In Section \ref{SectionZhangGeTheoremWithoutRestriction} we present an argument showing that the condition $k(n,e)$ be squarefree can be omitted in the Zhang-Ge criterion, that is, $n\in \mbox{ZG}(e)$ is a sufficient condition to guarantee $\lp{e}=\emptyset$ (Theorem \ref{ThMain1}). This result implies that $\lp{3}=\emptyset$ and $\lp{4}=\emptyset$ for infinitely many dimensions $n$ (Corollaries \ref{CorE3E4} and \ref{CorZG3ZG4}). Then, we proved that if the base-$3$ representation of $e$ contains either no digit $1$ or a unique digit $1$ which is in the unit place, then $\mbox{ZG}(e)=\emptyset$ (Propositions \ref{PropZGempty} and \ref{PropZGempty2}); and this is, in fact, the only cases where this happens (Theorem \ref{ThClassificationOfZG}). To finish Section \ref{SectionZhangGeTheoremWithoutRestriction} we derive some congruences for $k(n,e)$ and other related quantity $p(n,e)$ (Proposition \ref{PropPeriodicityForkANDp}) which are used in the next sections and prove that (under certain conditions) if the Zhang-Ge set $\mbox{ZG}(e)$ is non-empty, it contains infinitely many elements (Corollary \ref{CorEmptyOrInfinite}). In Section \ref{SectionMain} we obtain a classification of the Zhang-Ge sets (Theorem \ref{ThClassificationOfZG}) and prove one of the main result of this paper: there is a density-$1$ subset $E\subseteq \Z^{+}$ (containing $e=3$ and $e=4$) such that for every $e\in E$ we have that $\lp{e}=\emptyset$ for infinitely many values of $n$ (Theorem \ref{ThMain2} and Proposition \ref{PropEhasDensity1}). In Section \ref{SectionBeyond}, for each $k\geq 1$, we associate with each finite subset $B\subseteq \Z^n$, an homogeneous polynomial of degree $2k$: the $Q$-polynomial of $B$. Some explicit formulas for the $Q$-polynomials as a linear combination of monomial symmetric functions (Proposition \ref{PropFormulaForQB-monomial}) and linear combination of power sum symmetric functions (Proposition \ref{PropFormulaFinalForQB}) are derived. The later is used to obtain a general criterion for the non-existence of certain lattice tilings (Theorem \ref{TheoremMainGeneralized}). This criterion can be applied to prove not only non-existence results in the Lee metric but also to other types of metrics such as the $\ell_p$-metrics which are also of interest \cite{CJSC16} (see Example \ref{Example-Euclidean} for an application to the Euclidean metric). A specialization of this criterion to the perfect Lee codes is given in Proposition \ref{Prop-Specialization-for-Lee}: the $p$-condition of non-existence, where $p$ is an arbitrary odd prime. The $3$-condition of non-existence is equivalent to the Zhang-Ge condition which was studied in the first sections of this paper. Other choices of $p$ provide new non-existence criteria for perfect Lee codes. For instance, using the $5$-condition of non-existence we can extend Theorem \ref{ThMain2} to other radii such as $e=2, 6$ and $7$ (i.e. for these values of $e$, no linear perfect $e$-error-correcting Lee codes exist for infinitely many dimensions $n$).


%
%
%




\section{The Zhang-Ge sets and some congruences for $k(n,e)$ and $p(n,e)$}\label{SectionZhangGeTheoremWithoutRestriction}

\subsection{An extension of the Zhang-Ge theorem without the squarefree restriction}

Let $B^n(e)=\{x\in \Z^n: d(x,0)\leq e\}$ and $k(n,e)=\# B^n(e)$. It is well known \cite{GW70} that $k(n,e)=\sum_{i=0}^{\min\{n,e\}} 2^i \binom{n}{i}\binom{e}{i}$. Since $\binom{n}{i}=0$ for $i>n$ and $\binom{e}{i}=0$ for $i>e$ we can also write 
\begin{equation}\label{EqForkne1}
k(n,e)=\sum_{i=0}^{N} 2^i \binom{n}{i}\binom{e}{i}
\end{equation}
for any $N\geq \min\{n,e\}$ (in this paper we usually choose $N=3^m-1$ for large enough $m$). Let $p(n,e) = \sum_{i=1}^{e}2^i \sum_{j=1}^{e-i+1}j^2 \binom{e-j}{i-1}\binom{n-1}{i-1}$. Rearranging the sum and using that $\binom{a}{b}=0$ for $b>a$ we obtain the expression 
\begin{equation}\label{EqFormulaForp}
p(n,e) = \sum_{i=0}^{e}2i^2 k(n-1,e-i)
\end{equation} 
for every $n,e \geq 1$. In other words $p(n,e)$ is the coefficient of $x^e$ of the convolution of the generating functions $f(x)=\sum_{i=0}^{\infty} 2i^2 x^i$ and $g(x)=\sum_{i=0}^{\infty} k(n-1,i)x^i$. This observation is useful in order to obtain a generalization of the Zhang-Ge theorem (Theorem \ref{ZhanGeTheorem}) in Section \ref{SectionBeyond}.

We say that a pair of positive integers $(n,e)$ satisfies the \emph{Zhang-Ge condition} if it verifies the following system 
\begin{equation}\label{EqZhangGesystem}
\left\{ \begin{array}{l}
k(n,e)\equiv 3 \textrm{ or } 6 \pmod{9}, \\ p(n,e)\equiv 0 \pmod{3}.
\end{array} \right.
\end{equation}

\begin{definition}
The Zhang-Ge set of $e\geq 1$ is the set $$\mbox{ZG}(e)=\{n\geq 1: (n,e)\textrm{ satisfies the Zhang-Ge condition}\}.$$ We also denote the set of all linear $e$-perfect Lee codes $C\subseteq \Z^n$ by $\mbox{LPL}(n,e)$.
\end{definition}

The following result of T. Zhang and G. Ge establishes a necessary condition for the non-existence of perfect Lee codes.

\begin{theorem}[{\cite[Theorem~7]{ZG17}}]\label{ZhanGeTheorem}
If $n\in \mbox{ZG}(e)$ and $k(n,e)$ is squarefree, then $\mbox{LPL}(n,e)=\emptyset$
\end{theorem}

\begin{remark} In \cite{ZG17} the authors impose the extra condition $n\geq e$ in the statement of the theorem above. However this condition is not used in their proof and the result is valid also for $n<e$. A possible reason is because, in some corollaries of this theorem, they used Equation (\ref{EqForkne1}) for $k(n,e)$ with $N=e$. However, as mentioned above, the formula $k(n,e)=\sum_{i=0}^{e} 2^i \binom{n}{i}\binom{e}{i}$ holds also for $e>n$.
\end{remark}

For the case $e=3$ and $e=4$, we have the following corollary:

\begin{corollary}[{\cite[Corollaries~8 and 9]{ZG17}}]\label{ZhanGenCorollaries}
If $k(n,3)$ is squarefree and $n\equiv 12 \textrm{ or } 21 \pmod{27}$ then $\lp{3}=\emptyset$. If $k(n,4)$ is squarefree and $n\equiv 3,5,21 \textrm{ or } 23 \pmod{27}$ then $\lp{4}=\emptyset$.
\end{corollary}


Next we present an argument to show that the squarefree condition in Theorem \ref{ZhanGeTheorem} and Corollary \ref{ZhanGenCorollaries} can be skipped. We start by  stating one of the main tool to prove the non-existence of linear perfect Lee codes. 

\begin{theorem}[{\cite[Theorem~6]{HA12a}}]\label{ThHorakCriterion}
Let $B$ be a subset of $\Z^n$. Then, there is a lattice
tiling of $\Z^n$ by $B$ if and only if there is an abelian group
$G$ of order $|B|$ and a homomorphism $\phi:\Z^n \rightarrow G$ such that the restriction of $\phi$ to $B$ is a bijection.
\end{theorem}

\begin{corollary}
$ \lp{e}\neq \emptyset$ if and only if there is an abelian group $G$ and a homomorphism $\phi : \Z^n \rightarrow G$ such that $\phi|_{B^{n}(e)}: B^n(e)\rightarrow G$ is a bijection.
\end{corollary}

As in the proof of \cite[Theorem~3]{QCC18}, the main idea to obtain a version of Theorem \ref{ZhanGeTheorem} without the squarefree condition is to compose the homomorphism given in Theorem \ref{ThHorakCriterion} with a suitable homomorphism $\psi: G \rightarrow \Z_p$ for some prime $p$. The following lemma is a direct consequence of the structure theorem for finite abelian groups.

\begin{lemma}\label{LemmaEasy}
Let $G$ be an abelian group. If $|G|=pm$ with $\gcd(p,m)=1$, then there is an onto homomorphism $\psi:G\rightarrow \Z_p$. In particular, $\psi$ is an $m$-to-$1$ map.
\end{lemma}

\begin{theorem}\label{ThMain1}
If $n\in \mbox{ZG}(e)$ then $\lp{e}=\emptyset$.
\end{theorem}

\begin{proof}
Since $k(n,e)\equiv 3\textrm{ or }6\pmod{9}$, we can write $k(n,e)=3m$ with $m\in \Z^{+}$ and $3\nmid m$. Now we assume, by contradiction, that $\lp{e}\neq \emptyset$. By Theorem \ref{ThHorakCriterion}, there exist an abelian group $G$ of order $k(n,e)$ and a homomorphism $\phi:\Z^n \rightarrow G$ such that its restriction $\phi|_{B^n(e)}: B^n(e)\rightarrow G$ is bijective. By Lemma \ref{LemmaEasy}, there is an $m$-to-$1$ homomorphism $\psi:G\rightarrow \Z_3$. Then, the composition $f=\psi\circ\phi : \Z^n \rightarrow \Z_3$ is a homomorphism verifying that its resctriction $f|_{B^n(e)}: B^n(e)\rightarrow \Z_{3}$ is an $m$-to-$1$ map. We denote by $a_i=f(e_i)$ where $\{e_1,\ldots,e_n\}$ is the standard basis for $\Z^n$. For each value of $s \in \{0,1,2\}$ we have exactly $m$ values of $b=(b_1,\ldots, b_n)\in B^n(e)$ such that $f(b)=\sum_{i=1}^{n}b_ia_i \equiv s \pmod{3}$. Therefore 
$$\sum_{b\in B^n(e)} \left(\sum_{i=1}^{n}b_ia_i\right)^{2} \equiv m\cdot(0^2+1^2+2^2)\equiv -m \pmod{3}. $$ The first sum equals $p(n,e)\cdot \left(\sum_{i=1}^{n}a_i^2 \right)$ \cite[Theorem 7]{ZG17}. Then, it is a multiple of $3$ because $p(n,e)\equiv 0 \pmod{3}$, but $m$ is not, which is a contradiction.
\end{proof}

In the same way as the authors of \cite{ZG17} obtained Corollary \ref{ZhanGenCorollaries} from Theorem \ref{ZhanGeTheorem}, the following corollary can be obtained from Theorem \ref{ThMain1}.

\begin{corollary}\label{CorE3E4}
If $n\equiv 12\textrm{ or } 21 \pmod{27}$ then $\lp{3}=\emptyset$. If $n\equiv 3,5,21\textrm{ or }23 \pmod{27}$ then $\lp{4}=\emptyset$.
\end{corollary}

\begin{corollary}\label{CorZG3ZG4}
The Zhang-Ge sets $\mbox{ZG}(3)$ and $\mbox{ZG}(4)$ have infinitely many elements.
\end{corollary}

The first goal is to determine when the Zhang-Ge set $\mbox{ZG}(e)$ is either empty or non-empty. For those values of $e$ for which this set is non-empty, we also want to determine if this set is either finite or infinite.

\subsection{Cases where the Zhang-Ge set $ZG(e)$ is empty}

Here we apply classical Lucas' theorem on binomial coefficients to show some cases where $ZG(e)=\emptyset$. Lucas' theorem states that if $p$ is a prime number, $a=\sum_{j=0}^{h-1}a_j p^j$ with $a_j\in\{0,1,2\}$ and $b=\sum_{j=0}^{h-1}b_j p^j$ with $b_j\in\{0,1,2\}$, then $\binom{a}{b}\equiv \prod_{j=0}^{h-1}\binom{a_j}{b_j}\pmod{p}$. First we prove that if the base-$3$ representation of $e$ does not contain a digit $1$ then $ZG(e)=\emptyset$. The following lemma shows an important multiplicative property of $k(n,e)$.

\begin{lemma}\label{LemmakIsMultiplicativeMod3}
Let $n=\sum_{j=0}^{h-1}n_j\cdot 3^{j}$ and $e=\sum_{j=0}^{h-1}e_{j}\cdot 3^{j}$ with $n_j\in\{0,1,2\}$ and $e_{j}\in \{0,1,2\}$. Then $k(n,e)  \equiv \prod_{j=0}^{h-1}k(n_j,e_j) \pmod{3}$.
\end{lemma}

\begin{proof}
 Using Lucas' Theorem and $(-1)^{3^j}=-1$ for all $j\geq 0$ we have:
\begin{align*}
k(n,e) &=\sum_{i=0}^{3^h-1}2^i\binom{n}{i}\binom{e}{i} \equiv \sum_{i=0}^{3^h-1}(-1)^{i}\binom{n}{i}\binom{e}{i} 
  \equiv \sum_{i_{h-1}=0}^{2}\cdots\sum_{i_0=1}^{2}\left\{ \prod_{j=0}^{h-1}(-1)^{i_j}\binom{n_j}{i_j}\binom{e_j}{i_j}\right\}\\
   & = \prod_{j=0}^{h-1}\left(\sum_{i=0}^{2} (-1)^{i}\binom{n_j}{i}\binom{e_j}{i} \right) \equiv \prod_{j=0}^{h-1}\left(\sum_{i=0}^{2} 2^{i}\binom{n_j}{i}\binom{e_j}{i} \right) \equiv \prod_{j=0}^{h-1} k(n_j,e_j) \pmod{3}.
\end{align*}
\end{proof}

\begin{proposition}\label{PropZGempty}
If the base-$3$ representation of $e$ contains no digit $1$, then $k(n,e)\not\equiv 0 \pmod{3}$ for all $n\geq 1$. In particular, $ZG(e)=\emptyset$.
\end{proposition}

\begin{proof}
Let $n=\sum_{j=0}^{h-1}n_j\cdot 3^{j}$ and $e=\sum_{j=0}^{h-1}e_{j}\cdot 3^{j}$ with $n_j\in\{0,1,2\}$ and $e_{j}\in \{0,2\}$. By Lemma \ref{LemmakIsMultiplicativeMod3}, to prove that $3\nmid k(n,e)$ it suffices to prove that $3\nmid k(n_j,e_j)$ for every $j$, $0\leq j <h$. There are two cases to consider. If $e_j=0$ then $k(n_j,e_j)=1\not\equiv 0 \pmod{3}$ and if $e_j=2$ then $k(n_j,e_j)=2n_j^2+2n_j+1\not\equiv 0 \pmod{3}$ for $n_j=0,1,2$.
\end{proof}

Next we prove that if the base-$3$ representation of $e$ contains exactly one digit $1$ and it is in the unit place then $ZG(e)=\emptyset$. In Section \ref{SectionMain} we prove that these two cases are the only ones for which it happens. We start with some preliminary lemmas.

\begin{lemma}\label{LemmaFormulaForp}
Let $n$, $i$, $h$ and $s$ be positive integers such that $s\geq h$, $n-1=\sum_{j=0}^{s-1}n_j3^j$ where every $n_j\in\{0,1,2\}$, and $i<3^{h-1}$. Then
\begin{itemize}
\item[i)] $k(n-1,2\cdot 3^{h-1}+i)\equiv (-1)^{n_{h-1}}\cdot k(n-1,i)\pmod{3}$.
\item[ii)] $k(n-1, 3^{h-1}+i)\equiv (1-n_{h-1})\cdot k(n-1,i)\pmod{3}$.
\end{itemize}
\end{lemma}

\begin{proof}
By Lemma \ref{LemmakIsMultiplicativeMod3} and the fact that $k(n_j,0)=1$ for every $j$, we have
$k(n-1, 2 \cdot 3^{h-1} + i) \equiv k(n_{h-1},2)\cdot k(n-1,i)\pmod{3}$ and $k(n-1, 3^{h-1} + i) \equiv k(n_{h-1},1)\cdot k(n-1,i)\pmod{3}$. The conclusion follows from the fact that $k(n_{h-1},1)=2n_{h-1}+1\equiv 1-n_{h-1} \pmod{3}$ and $k(n_{h-1},2)=2n_{h-1}^2+2n_{h-1}+1\equiv (-1)^{n_{h-1}}\pmod{3}$ for $0\leq n_{h-1}\leq 2$.
\end{proof}

\begin{lemma}\label{LemmaOtherFormulaForp}
Let $h>1$, $n\equiv \sum_{i=0}^{h-1}n_i 3^i+1 \pmod{3^h}$ with $n_i\in \{0,1,2\}$, and $e=2\cdot 3^{h-1} + e'$ with $0\leq e'<3^{h-1}$. If $n_{h-1}=2$ or $n_i\neq 2$ for some $i: 1\leq i \leq h-2$ then $$p(n,e)\equiv (-1)^{n_{h-1}}\cdot p(n,e')\pmod{3}.$$ If $n_{h-1}\neq 2$ and $n_i=2$ for each $i: 1\leq i \leq h-2$ then 
$$p(n,e)\equiv (-1)^{n_{h-1}}\cdot p(n,e')+ 2(2-n_{h-1})\!\!\!\!\!\!\doblesum{0\leq i < 3}{i\not\equiv e\!\!\!\!\! \pmod{3}}\!\!\!\!\! k(n-1, i)\pmod{3}.$$
\end{lemma}

\begin{proof}
Using Equation (\ref{EqFormulaForp}) we write $p(n,e)=\sum_{i=0}^{e}2i^2 k(n-1.e-i)$. Since $i^2 \equiv 1 \pmod{3}$ when $i\not\equiv 0 \pmod{3}$, we have
\begin{equation}\label{EqFormulaForpMod3}
p(n,e)\equiv 2\cdot\!\!\!\!\!\!  \doblesum{0\leq i \leq e}{i\not\equiv 0\!\!\!\!\! \pmod{3}}\!\!\!\!\! k(n-1,e-i)\equiv 2\cdot\!\!\!\!\!\!  \doblesum{0\leq i \leq e}{i\not\equiv e\!\!\!\!\! \pmod{3}}\!\!\!\!\!  k(n-1,i) \pmod{3}.
\end{equation}
We split this sum in three parts
$$ p(n,e)\equiv  2\cdot\!\!\!\!\!\!\doblesum{2\cdot 3^{h-1}\leq i \leq e}{i\not\equiv e\!\!\!\!\! \pmod{3}}\!\!\!\!\!  k(n-1,i) + 2\cdot\!\!\!\!\!\! \doblesum{3^{h-1}\leq i < 2\cdot 3^{h-1}}{i\not\equiv e\!\!\!\!\! \pmod{3}}\!\!\!\!\!k(n-1,i) + 2\cdot\!\!\!\!\!\!\doblesum{0\leq i < 3^{h-1}}{i\not\equiv e\!\!\!\!\! \pmod{3}}\!\!\!\!\!   k(n-1,i) \pmod{3}.$$

By Equation (\ref{EqFormulaForpMod3}), Lemma \ref{LemmaFormulaForp} and the fact that $e\equiv e'\pmod{3}$, for the first sum we have
$$2\cdot\!\!\!\!\!\!\doblesum{2\cdot 3^{h-1}\leq i \leq e}{i\not\equiv e\!\!\!\!\! \pmod{3}}\!\!\!\!\! k(n-1,i) \equiv (-1)^{n_{h-1}}\cdot 2\cdot\!\!\!\!\!\!  \doblesum{0\leq i \leq e'}{i\not\equiv e\!\!\!\!\! \pmod{3}}\!\!\!\!\! k(n-1, i) =(-1)^{n_{h-1}}\cdot p(n,e') \pmod{3},$$ and for the second sum we have
$$\doblesum{3^{h-1}\leq i < 2\cdot 3^{h-1}}{i\not\equiv e\!\!\!\!\! \pmod{3}}\!\!\!\!\! k(n-1,i) \equiv (1-n_{h-1})\cdot\!\!\!\!\!\!  \doblesum{0\leq i < 3^{h-1}}{i\not\equiv e\!\!\!\!\! \pmod{3}}\!\!\!\!\! k(n-1, i) \pmod{3}. $$
Thus,
\begin{equation}\label{EqAlfa}
p(n,e) \equiv   (-1)^{n_{h-1}}\cdot p(n,e') + 2(2-n_{h-1})\cdot\!\!\!\!\!\!  \doblesum{0\leq i < 3^{h-1}}{i\not\equiv e\!\!\!\!\! \pmod{3}}\!\!\!\!\! k(n-1, i) \pmod{3}.
\end{equation}
Note that for $t>1$ we have:
\begin{align}\label{EqBeta}
\doblesum{0\leq i < 3^{t}}{i\not\equiv e\!\!\!\!\! \pmod{3}}\!\!\!\!\! k(n-1, i) &= \!\!\!\!\!\!  \doblesum{0\leq i < 3^{t-1}}{i\not\equiv e\!\!\!\!\! \pmod{3}}\!\!\!\!\! k(n-1, 2\cdot 3^{t-1}+i) +   \!\!\!\!\!\!  \doblesum{0\leq i < 3^{t-1}}{i\not\equiv e\!\!\!\!\! \pmod{3}}\!\!\!\!\! k(n-1,3^{t-1}+i) + \!\!\!\!\!\!  \doblesum{0\leq i < 3^{t-1}}{i\not\equiv e\!\!\!\!\! \pmod{3}}\!\!\!\!\! k(n-1,  i) \nonumber  \\
&\equiv  (-1)^{n_{t-1}} \!\!\!\!\!\!  \doblesum{0\leq i < 3^{t-1}}{i\not\equiv e\!\!\!\!\! \pmod{3}}\!\!\!\!\! k(n-1, i) + 
(1-n_{t-1})\!\!\!\!\!\!  \doblesum{0\leq i < 3^{t-1}}{i\not\equiv e\!\!\!\!\! \pmod{3}}\!\!\!\!\! k(n-1, i) + \!\!\!\!\!\!  \doblesum{0\leq i < 3^{t-1}}{i\not\equiv e\!\!\!\!\! \pmod{3}}\!\!\!\!\! k(n-1, i) \nonumber \\
& \equiv (2-n_{t-1}+(-1)^{n_{t-1}})\!\!\!\!\!\!  \doblesum{0\leq i < 3^{t-1}}{i\not\equiv e\!\!\!\!\! \pmod{3}}\!\!\!\!\! k(n-1, i)\pmod{3} 
\end{align}
Using Equations (\ref{EqAlfa}) and (\ref{EqBeta}) with $t=h-1,h-2,\ldots,1$ we obtain:

\begin{equation}\label{EqGamma}
\doblesum{0\leq i < 3^{h-1}}{i\not\equiv e\!\!\!\!\! \pmod{3}}\!\!\!\!\! k(n-1, i) \equiv  \prod_{i=1}^{h-2} (2-n_{i}+(-1)^{n_{i}})\cdot \!\!\!\!\!\!  \doblesum{0\leq i < 3}{i\not\equiv e\!\!\!\!\! \pmod{3}}\!\!\!\!\! k(n-1, i), 
\end{equation}
where the product above is $1$ when $h=2$. Combining Equations (\ref{EqAlfa}) and (\ref{EqGamma}) we obtain:
$$p(n,e) \equiv   (-1)^{n_{h-1}}\cdot p(n,e') + 2(2-n_{h-1})\cdot\prod_{i=1}^{h-2} (2-n_{i}+(-1)^{n_{i}})\cdot \!\!\!\!\!\!  \doblesum{0\leq i < 3}{i\not\equiv e\!\!\!\!\! \pmod{3}}\!\!\!\!\! k(n-1, i) \pmod{3}.    $$
The conclusion follows from the fact that $\prod_{i=1}^{h-2} (2-n_{i}+(-1)^{n_{i}}) \equiv 1 \pmod{3}$ if $n_1=n_2=\cdots=n_{h-2}=2$ and $\prod_{i=1}^{h-2} (2-n_{i}+(-1)^{n_{i}}) \equiv 0 \pmod{3}$ otherwise.
\end{proof}

\begin{lemma}\label{Lemmap2mod3}
Let $h>1$, $n\equiv \sum_{i=0}^{h-1}n_i 3^i+1 \pmod{3^h}$ with $n_0\in\{0,1,2\}$ and $n_1=\cdots=n_{h-1}=2$. Let $e= \sum_{i=0}^{h-1}e_i 3^i$ with $e_0=1$ and $e_i\in\{0,2\}$ for $1\leq i \leq h-1$. Then $p(n,e)\equiv 2 \pmod{3}$.
\end{lemma}

\begin{proof}

Let $\mathcal{E}=\{1+\sum_{i=1}^{h-1}e_i 3^i : h\geq 0, e_1, \cdots, e_{h-1} \in \{0,2\} \}$. We prove this lemma by induction on $e \in \mathcal{E}$. If $e=1$ then $p(n,1)=2k(n-1,0)=2$. Now we suppose that $e\in \mathcal{E}$ with $e>1$ and write $e=2\cdot 3^{t-1}+e'$ with $1<t\leq h$ and $0\leq e' < 3^{t-1}$. It is clear that $e'\in \mathcal{E}$ and $e'<e$. Thus, by inductive hypothesis we have $p(n,e')\equiv 2 \pmod{3}$. Since $n-1= \sum_{i=0}^{t-1}n_i 3^i \pmod{3^t}$ with $n_{t-1}=2$ (because $t\leq h$) and $n_i \in \{0,1,2\}$ for $0\leq i <t-1$, we can use Lemma \ref{LemmaOtherFormulaForp} together with the inductive hypothesis to obtain:
$$ p(n,e)\equiv (-1)^{n_{t-1}}p(n,e')= p(n,e') \equiv 2 \pmod{3}. $$
\end{proof}


\begin{proposition}\label{PropZGempty2}
If the base-$3$ representation of $e$ contains exactly one digit $1$ and it is in the unit place then $p(n,e)\not\equiv 0 \pmod{3}$ for all $n\geq 1$. In particular $ZG(e)=\emptyset$.
\end{proposition}

\begin{proof} 

Let $n$ be a positive integer and $\mathcal{E}$ be the subset of positive integers whose base-$3$ representation contains exactly one digit $1$ and it is in the unit place. Let $n-1=\sum_{i=0}^{\infty}n_i 3^i$ be the base-$3$ representation of $n-1$ with each $n_i \in \{0,1,2\}$. To prove this lemma we proceed again by induction on $e\in \mathcal{E}$. If $e=1$ we have $p(n,e)=2\not\equiv 0 \pmod{3}$. Now we consider $e\in \mathcal{E}$ with $e>1$ and write $e=2\cdot 3^{h-1}+e'$ with $0\leq e' < 3^{h-1}$ and $h>1$. It is clear that $e'\in \mathcal{E}$ and $e'<e$. Thus, by inductive hypothesis we can assume $p(n,e')\not\equiv 0 \pmod{3}$. We note that $n-1\equiv \sum_{i=0}^{h-1}n_i 3^i \pmod{3^h}$ with each $n_i\in\{0,1,2\}$ and consider two cases. If $n_{h-1}=2$ or $n_i\neq 2$ for some $i\in\{1,\ldots, h-2\}$ we apply Lemma \ref{LemmaOtherFormulaForp} to obtain $p(n,e)\equiv (-1)^{n_{h-1}} p(n,e')\not\equiv 0 \pmod{3}$. If $n_{h-1}\in\{0,1\}$ and $n_{1}=\cdots= n_{h-2}=2$, by Lemma \ref{Lemmap2mod3} we have $p(n,e')\equiv 2 \pmod{3}$ and by Lemma \ref{LemmaOtherFormulaForp} we have:
\begin{align*}
 p(n,e)\equiv & (-1)^{n_{h-1}}\cdot p(n,e')+ 2(2-n_{h-1})(k(n-1,0)+k(n-1,2))\\ \equiv &  (-1)^{n_{h-1}}\cdot p(n,e')+(-1)^{n_{h-1}}(1+2n^2-2n+1)\\ =& (-1)^{n_{h-1}} (p(n,e')+2n^2-2n+2)\\ \equiv & (-1)^{n_{h-1}} (2n_0^2 -2 n_0+1) \not\equiv 0 \pmod{3}.
\end{align*} 

\end{proof}

By convenience, we define the following function.

\begin{definition}\label{DefDelta3}
For $n\geq 1$ we consider its base-$3$ representation $n=\sum_{i=0}^{h-1} n_i 3^i$ with each $n_i\in\{0,1,2\}$. The function $\delta_3 : \Z^{+}  \rightarrow \N \cup\{\infty\}$ is given by $$\delta_3(n)=\left\{ \begin{array}{ll}
\max\{i: n_i=1\} & \textrm{if $n_i=1$ for some $i\geq 0$},\\
\infty & \textrm{if $n_i\neq 1$ for all $i\geq 0$}.
\end{array} \right.$$
\end{definition}

The following corollary is a consequence of Propositions \ref{PropZGempty} and \ref{PropZGempty2}.

\begin{corollary}\label{CorZempty}
Let $e\geq 1$. If $\delta_3(e)=0$ or $\delta_3(e)=\infty$ then $\mbox{ZG}(e)=\emptyset$.
\end{corollary}

\subsection{The Davis-Webb theorem and congruence formulas for $k(n,e)$ and $p(n,e)$}

We proved that $\mbox{ZG}(e)=\emptyset$ if $\delta_3(e)\in\{0,\infty\}$. In the next section we prove that $\mbox{ZG}(e)\neq\emptyset$ if $0<\delta_3(e)<\infty$. One ingredient of the proof is a generalization of Lucas' Theorem on binomial coefficients given by Davis and Webb in \cite{DW90}. For $p$ prime and $0\leq a,b <p$ the Davis-Webb symbol is defined by $\simb{a}{b}=\binom{a}{b}$ if $a\geq b$ and $\simb{a}{b}=p$ if $a<b$ (we note that in this case $\binom{a}{b}=0$). For $0\leq a,b,c,d <p$ the Davis-Webb symbol is defined by $\simb{a,b}{c,d}=\binom{ap+b}{cp+d}$ if $ap+b\geq cp+d$ and $\simb{a,b}{c,d}=p\simb{b}{d}$ otherwise.

\begin{theorem}[{\cite[Theorem 1]{DW90}}]\label{DavisWebbTheorem} Let $p$ be a prime number and $a\geq b$ be natural numbers. If $a=\sum_{i=0}^{m-1}a_ip^i$ and $b=\sum_{i=0}^{m-1}b_i p^i$ with $m\geq 2$, $0\leq a_i,b_i<p$ and $a_{m-1}> 0$, then 
\begin{equation}\label{EqDavisWebb}
\binom{a}{b}\equiv    \simb{a_{m-1},a_{m-2}}{b_{m-1},b_{m-2}}\cdot   \left( \prod_{i=1}^{m-2}\simb{a_{i}}{b_{i}}^{-1} \cdot  \simb{a_{i},a_{i-1}}{b_{i},b_{i-1}}\right) \pmod{p^2}.
\end{equation}
\end{theorem}

\begin{remark}
If $0\leq b_i \leq a_i <p$, then $\simb{a_{i}}{b_{i}}=\binom{a_{i}}{b_{i}}$ is coprime with $p$ and $\simb{a_{i}}{b_{i}}^{-1} \cdot  \simb{a_{i},a_{i-1}}{b_{i},b_{i-1}}$ has no $p$ in the denominator. If $0\leq a_i < b_i <p$, then $\simb{a_{i}}{b_{i}}^{-1} \cdot  \simb{a_{i},a_{i-1}}{b_{i},b_{i-1}} = \simb{a_{i-1}}{b_{i-1}}$ is an integer. Therefore, by Equation (\ref{EqDavisWebb}), if $\simb{a_{m-1},a_{m-2}}{b_{m-1},b_{m-2}}=p^2$ or $\simb{a_{i}}{b_{i}}^{-1} \cdot  \simb{a_{i},a_{i-1}}{b_{i},b_{i-1}}\equiv 0 \pmod{p^2}$ for some $i\in\{1,\ldots,m-2\}$ then $\binom{a}{b}\equiv 0 \pmod{p^2}$.
\end{remark}

In general Equation (\ref{EqDavisWebb}) does not hold when $a<b$. For example if $a=p^2+p+1$ and $b=2p^2+p+1$ then the right hand side of Equation  (\ref{EqDavisWebb}) is $\simb{1,1}{2,1} \cdot \simb{1}{1}^{-1}\simb{1,1}{1,1} = p$ but $\binom{a}{b}=0$. The next simple result will be used in the next section and it is true even if $a<b$.

\begin{proposition}\label{PropDavisWebbExtension}
Let $p$ be a prime number, $a=\sum_{i=0}^{m-1}a_i p^i$ and $b=\sum_{i=0}^{m-1}b_i p^i$ with $m\geq 2$, $0\leq a_i,b_i<p$ and $a_{m-1}=b_{m-1}$, then Equation (\ref{EqDavisWebb}) holds. 
\end{proposition}

\begin{proof}
By Theorem \ref{DavisWebbTheorem}, it is enough to prove that 
\begin{equation}\label{Eqa}
\simb{a_{m-1},a_{m-2}}{a_{m-1},b_{m-2}}\cdot   \left( \prod_{i=1}^{m-2}\simb{a_{i}}{b_{i}}^{-1} \cdot  \simb{a_{i},a_{i-1}}{b_{i},b_{i-1}}\right) \equiv 0 \pmod{p^2}
\end{equation} whenever $a<b$. If $b_{m-2}>a_{m-2}$ we have $\simb{a_{m-1},a_{m-2}}{a_{m-1},b_{m-2}}=p^2$ and Equation (\ref{Eqa}) holds. Otherwise, there is an integer $i$, $1\leq i \leq m-2$ such that $a_{m-1}=b_{m-1}, a_{m-2}=b_{m-2},\ldots, a_{i}=b_{i}$ and $a_{i-1}<b_{i-1}$. In this case we have $\simb{a_{i}}{b_{i}}^{-1} \cdot  \simb{a_{i},a_{i-1}}{b_{i},b_{i-1}} = \simb{a_{i},a_{i-1}}{a_{i},b_{i-1}} = p^2$ and Equation (\ref{Eqa}) holds.
\end{proof}

\begin{example}
Let $p$ be a prime number, $a=p^2+p+1$ and $b=2p^2+p+1$. Since $a=\sum_{i=0}^{3}a_ip^i$ and $b=\sum_{i=0}^{3}b_i p^i$ with $a_3=b_3=0$ (and $a_2=a_1=a_0=b_1=b_0=1, b_2=2$) we can apply Proposition \ref{PropDavisWebbExtension} to obtain $\binom{a}{b}\equiv \simb{0,1}{0,2} \cdot \simb{1}{2}^{-1}\simb{1,1}{2,1} \cdot \simb{1}{1}^{-1}\simb{1,1}{1,1} \pmod{p^2}$.
\end{example}

Next we apply Theorem \ref{DavisWebbTheorem} and Proposition \ref{PropDavisWebbExtension} to deduce some congruences for $k(n,e)$ and $p(n,e)$. Then, we use these congruences to prove that if a Zhang-Ge set $\mbox{ZG}(e)$ contains an element of special type, it contains infinitely many elements (Corollary \ref{CorEmptyOrInfinite}).

\begin{lemma}\label{LemmaBinomialPeriodicity}
Let $n,m,h$ and $a$ be positive integers satisfying $h\geq m+2$ and  $3^{m+1}\leq n<2\cdot 3^{m+1}$. If $i<3^h$, then $\binom{3^{h+1}a+n}{i}\equiv \binom{n}{i} \pmod{9}$.
\end{lemma}

\begin{proof}
Let $a\in \Z^{+}$ and write $3^{h+1}a+n = \sum_{j=0}^{h+k}n_j 3^{j}$ and $i=\sum_{j=0}^{h+k}i_j 3^j$ with $i_j,n_j\in\{0,1,2\}$, for some $k\geq 1$ and $n_{h+k}\neq 0$. By our hypothesis we have $n_{m+1}=1$, $n_j=0$ for $m+2\leq j \leq h$ and $i_j=0$ for $h\leq j \leq h+k$. First we suppose $2\cdot 3^{m+1}\leq i <3^h$ (in particular $n<i$) and consider the set of indices $J=\{j: m+2\leq j\leq h-1, i_j\neq 0\}$. If the set $J$ is non-empty we consider $j=\max J$. Since $n_{j+1}=n_j=i_{j+1}=0$, we have $\simb{n_{j+1}}{i_{j+1}}^{-1}\cdot \simb{n_{j+1},n_{j}}{i_{j+1},\ i_{j}} = \simb{0,0}{0,i_{j}} =9$. Then, by Theorem \ref{DavisWebbTheorem} we conclude that $\binom{3^{h+1}a+n}{i}\equiv 0 = \binom{n}{i} \pmod{9}$. If the set $J$ is empty we have $i<3^{m+2}$ (because $i<3^{h}$). Since $n_{m+2}=i_{m+2}=0$, $n_{m+1}=1$ and $i_{m+1}=2$ (because $2\cdot 3^{m+1}\leq i< 3^{m+2}$), we have $\simb{n_{m+2}}{i_{m+2}}^{-1}\cdot \simb{n_{m+2},n_{m+1}}{i_{m+2},\ i_{m+1}} = \simb{0,1}{0,2}=9$ and by Theorem \ref{DavisWebbTheorem} we conclude that $\binom{3^{h+1}a+n}{i}\equiv 0 = \binom{n}{i} \pmod{9}$. Now we suppose $i< 2\cdot 3^{m+1}$, this is equivalent to $i_j=0$ for $j>m+1$ and $i_{m+1}\leq 1 = n_{m+1}$. Applying Theorem \ref{DavisWebbTheorem} we have 
\begin{align*}
\binom{3^{h+1}a+n}{i} \equiv & \simb{n_{h+k},n_{h+k-1}}{0,0} \left( \prod_{j=m+3}^{h+k-1} \simb{n_j}{0}^{-1} \simb{n_j, n_{j-1}}{0,0} \right)\cdot \simb{0}{0}^{-1} \simb{0,n_{m+1}}{0,i_{m+1}}  \\  & \cdot \simb{n_{m+1}}{i_{m+1}}^{-1} \simb{n_{m+1},n_m}{i_{m+1},i_{m}}\cdot \left(  \prod_{j=1}^{m}\simb{n_j}{i_j}^{-1} \simb{n_j,n_{j-1}}{i_j, i_{j-1}}\right) \pmod{9}.
\end{align*}
We have $\simb{n_{h+k},n_{h+k-1}}{0,0} \left( \prod_{j=m+3}^{h+k-1} \simb{n_j}{0}^{-1} \simb{n_j, n_{j-1}}{0,0} \right)\cdot \simb{0}{0}^{-1}=1$, because all its terms are equal to $1$. Since $0 \leq i_{m+1}\leq n_{m+1}\leq 1$, we have $\simb{0,n_{m+1}}{0,i_{m+1}}=\simb{n_{m+1}}{i_{m+1}} =1$. Thus 
$$ \binom{3^{h+1}a+n}{i} \equiv   \simb{n_{m+1},n_m}{i_{m+1},i_{m}}\cdot \left(  \prod_{j=1}^{m}\simb{n_j}{i_j}^{-1} \simb{n_j,n_{j-1}}{i_j, i_{j-1}}\right) \equiv \binom{n}{i} \pmod{9},$$ where in the last congruence we use Theorem \ref{DavisWebbTheorem} if $i\leq n$ or Proposition \ref{PropDavisWebbExtension} if $i>n$ (since in this case we have $i_{m+1}=n_{m+1}=1$).
\end{proof}

\begin{proposition}\label{PropPeriodicityForkANDp}
Let $m$ be a natural number and $n,e$ and $h$ be positive integers such that $3^{m+1}\leq n<2\cdot 3^{m+1}$, $e<3^h$ and $h\geq m+2$. Then, the following congruences hold for every $a\geq 1$:
\begin{itemize}
\item[(i)] $k(3^{h+1}a+n,e)\equiv k(n,e) \pmod{9}$;
\item[(ii)] $p(3^{h}a+n,e)\equiv p(n,e) \pmod{3}$. 
\end{itemize}  
\end{proposition}

\begin{proof}
By Lemma \ref{LemmaBinomialPeriodicity} (and using $e\leq 3^h-1$) we have $k(3^{h+1}a+n,e) = \sum_{i=0}^{3^h -1}2^i \binom{3^{h+1}a+n}{i}\binom{e}{i}  \equiv  \sum_{i=0}^{3^h -1} 2^i \binom{n}{i}\binom{e}{i} = k(n,e) \pmod{9}$, which proves {\sf (i)}. To prove {\sf (ii)} we use Equation (\ref{EqFormulaForpMod3}) and Lemma \ref{LemmakIsMultiplicativeMod3} to obtain:
\begin{align*}
p(3^{h}a+n,e) &\equiv  2\cdot\!\!\!\!\!\!  \doblesum{0\leq i \leq e}{i\not\equiv e\!\!\!\!\! \pmod{3}}\!\!\!\!\!  k(3^{h}a+n-1,3^{h}\cdot 0+i) \equiv   2\cdot\!\!\!\!\!\!  \doblesum{0\leq i \leq e}{i\not\equiv e\!\!\!\!\! \pmod{3}}\!\!\!\!\!  k(a,0)\cdot k(n-1,i)\\ &\equiv  2\cdot\!\!\!\!\!\!  \doblesum{0\leq i \leq e}{i\not\equiv e\!\!\!\!\! \pmod{3}}\!\!\!\!\!  k(n-1,i) \equiv  p(n,e) \pmod{3}.
\end{align*}

\end{proof}

\begin{corollary}\label{CorEmptyOrInfinite}
Let $e\geq 1$. If $\mbox{ZG}(e)\cap \{3^{m+1},3^{m+1}+1,\ldots, 2\cdot  3^{m+1}-1\}\neq \emptyset$ is non-empty for some $m\geq 0$, it has infinitely many elements. Moreover, if $n\in \mbox{ZG}(e)$ with $3^{m+1}\leq n < 2\cdot 3^{m+1}$ and $h=\max\{m+3,\lfloor \log_3(e) \rfloor +1\}$ then $3^{h}\cdot \N +n \subseteq \mbox{ZG}(e)$.
\end{corollary}

\section{Classification of the Zhang-Ge sets $\mbox{ZG}(e)$ and non-existence results for perfect Lee codes}\label{SectionMain}


In this section we prove that the Zhang-Ge set $\mbox{ZG}(e)$ contains infinitely many elements if $0<\delta_3(e)<\infty$ (see Definition \ref{DefDelta3}) and obtain a non-existence result for linear perfect Lee codes (Theorem \ref{ThMain2}). By Corollary \ref{CorEmptyOrInfinite} it suffices to prove it contains an element $n$ with $3^{m+1}\leq n < 2 \cdot 3^{m+1}$ for some $m\geq 0$. We start with the case $\delta_3(e)=1$.

\subsection{The case $\delta_3(e)=1$}

Note that $\delta_3(e)=1$ if and only if $e=a+3^2 b$ with $a\in\{3,4,5\}$ and $b\geq 0 $ with $\delta_3(b)=\infty$.

\begin{proposition}\label{PropDelta1p1}
Let $e=a+3^2 b$ with $a\in\{3,5\}$ and $b\geq 0$ satisfying $\delta_3(b)=\infty$. Then $12 \in \mbox{ZG}(e)$.
\end{proposition}

\begin{proof}
We have to prove that $k(12,e)\equiv 3\textrm{ or }6 \pmod{9}$ and $p(12,e)\equiv 0 \pmod{3}$. We note that $k(12,e)$ is a degree-$12$ polynomial with rational coefficients. Multiplying it by $12!$ we obtain the integer coefficients polynomial $f(e)=12!\cdot k(12,e) = 4096 e^{12} + 24576 e^{11} + 585728 e^{10} + 2703360 e^9 + 25479168 e^8 + 85966848 e^7 + 402980864 e^6 + 919142400 e^5 + 2188865536 e^4$ $+ 2940850176 e^3 + 3130103808 e^2 + 1799331840 e + 479001600$. Since $3^5\mid 12!$, $3^6\nmid 12!$, $3\cdot 12!\equiv 2\cdot 3^6 \pmod{3^7}$ and $6\cdot 12!\equiv 3^6 \pmod{3^7}$, we have that $k(12,e)\equiv 3\textrm{ or }6 \pmod{3^2}$ if and only if $f(e)\equiv 3^6 \textrm{ or } 2\cdot 3^6 \pmod{3^7}$. We note that $e\equiv 3\textrm{ or }5 \pmod{3^2}$ if and only if $e\equiv a+3^2 b\pmod{3^7}$ for some $a\in \{3,5\}$ and $b: 0\leq b < 3^5$. By a direct calculation we confirm that for these $486$ possible values of $a+3^2b$ we have $f(a+3^2 b )\equiv 3^6\textrm{ or } 2\cdot 3^6 \pmod{3^7}$. Thus, $k(12,e)\equiv 3\textrm{ or }6\pmod{9}$ whenever $e\equiv 3\textrm{ or }5\pmod{9}$. 
Since $p(12,3)=732\equiv 0 \pmod{3}$, $p(12,5)=45870\equiv 0 \pmod{3}$
and $p(12, 2\cdot 3^{h-1}+i)\equiv \pm p(12,i) \pmod{3}$ for $h>2$ and $0\leq i < 3^{h-1}$ (Lemma \ref{LemmaOtherFormulaForp}), we have that 
$p(12,e)\equiv 0 \pmod{3}$ for every $e$ of the form $e=3+9b$ or $e=5+9b$ with $b\geq 0$ and $\delta_3(b)=\infty$.
\end{proof}

\begin{proposition}\label{PropDelta1p2}
Let $e=4+3^2 b$ with $b\geq 0$ and $\delta_3(b)=\infty$. Then, $3 \in \mbox{ZG}(e)$.
\end{proposition}

\begin{proof}
We consider the polynomial $f(e)=3! \cdot k(3,e) = 8e^3 + 12e^2 + 16e + 6$ and note that $k(3,e)\equiv 3\textrm{ or } 6 \pmod{3^2}$ if and only if $f(e)\equiv 9\textrm{ or }18 \pmod{3^3}$. If $e=4+3^2 b$ with $b\geq 0$ and $\delta_3(b)=\infty$, then $e\equiv 4\textrm{ or }22 \pmod{3^3}$. Since $f(4)=774\equiv 18 \pmod{3^3}$ and $f(22)=91350\equiv 9 \pmod{3^3}$ we have that $f(e)\equiv 9\textrm{ or }18 \pmod{3^3}$ if $e\equiv 4\textrm{ or }22 \pmod{3^3}$. This implies that $k(3,e)\equiv 3\textrm{ or }6 \pmod{9}$ if $e\equiv 4\textrm{ or }22 \pmod{3^3}$. In particular $k(3,4+3^2b)\equiv 3\textrm{ or }6 \pmod{3^2}$ for every $b\geq 0$ satisfying $\delta_3(b)=\infty$. 
Since $p(3,4)=276\equiv 0 \pmod{3}$ and $p(3, 2\cdot 3^{h-1}+i)\equiv \pm p(3,i) \pmod{3}$ for $h>2$ and $0\leq i < 3^{h-1}$ (Lemma \ref{LemmaOtherFormulaForp}), we have that $p(3,4+9b)\equiv 0 \pmod{3}$ for every $b\geq 0$ satisfying $\delta_3(b)=\infty$.
\end{proof}

The following corollary is a direct consequence of Propositions \ref{PropDelta1p1} and \ref{PropDelta1p2} and Corollary \ref{CorEmptyOrInfinite}.

\begin{corollary}\label{CorZGnonemptyForDelta1}
If $\delta_3(e)=1$ then the Zhang-Ge set $\mbox{ZG}(e)$ has infinitely many elements.
\end{corollary}

\subsection{The case $\delta_3(e)=2$}

Note that $\delta_3(e)=2$ if and only if $e=a+3^3 b$ with $9\leq a<18$ and $b\geq 0 $ such that $\delta_3(b)=\infty$.

\begin{proposition}\label{PropDelta2}
Let $e=a+3^3 b$ with $9\leq a <18$, $b\geq 0$ and $\delta_3(b)=\infty$. Then, $12 \in \mbox{ZG}(e)$.
\end{proposition}

\begin{proof}
The proof is similar to the proof of Proposition \ref{PropDelta1p1} and we give only a sketch. Consider the polynomial $f(e)=12!\cdot k(12,e)=  4096 e^{12} + 24576 e^{11} +\cdots + 479001600$. If $e=a+3^3 b$ with $9\leq a <18$ and $\delta_3(b)=\infty$, then $e\equiv a\textrm{ or } 2\cdot 3^3 +a \pmod{3^4}$ (because $b\equiv 0\textrm{ or }2 \pmod{3}$). Let $A=\{a': 9\leq a'<18 \textrm{ or } 63\leq a' <72\}$. We have that $e\equiv a'\pmod{3^4}$ for some $a'\in A$ if and only if $e\equiv a'+3^4 b' \pmod{3^7}$ with $a'\in A$ and $0\leq b' < 3^3$. By direct calculation we check that $f(a'+3^4 b')\equiv 3^6$ or $ 2\cdot 3^6 \pmod{3^7}$ for every $a'\in A$ and $0\leq b' <3^3$. This implies that $f(e)\equiv 3^6$ or $2\cdot 3^6 \pmod{3^7}$ for every $e\equiv a'+3^4 b' \pmod{3^7}$ with $a'\in A$ and $0\leq b' < 3^3$. Thus $k(12,e)\equiv 3,6 \pmod{3^2}$ if $e\equiv a'\pmod{3^4}$ for some $a'\in A$. In particular $k(12,e)\equiv 3\textrm{ or }6 \pmod{9}$ for every $e=a+3^3 b$ with $9\leq a <18$ and $\delta_3(b)=\infty$. Since $p(12,a)\equiv 0 \pmod{3}$ for $9\leq a < 18$ and $p(12, 2\cdot 3^{h-1}+i)\equiv \pm p(12,i) \pmod{3}$ for $h>3$ and $0\leq i < 3^{h-1}$ (Lemma \ref{LemmaOtherFormulaForp}) we have that $p(12,a+3^3 b)\equiv 0\pmod{3}$ if $9\leq a<18$ and $\delta_3(b)=\infty$.
\end{proof}

\begin{corollary}\label{CorZGnonemptyForDelta2}
If $\delta(e)=2$ then the Zhang-Ge set $\mbox{ZG}(e)$ has infinitely many elements.
\end{corollary}

\subsection{The case $2< \delta_3(e)<\infty$}

The arguments used in the proofs of Propositions \ref{PropDelta1p1}, \ref{PropDelta1p2} and \ref{PropDelta2} require intermediate computations. So, it becomes infeasible when $n$ is large. Thus we need a new argument to approach the case $2< \delta_3(e)<\infty$. In this part we consider $e\geq 0$ such that $\delta_3(e)=m+1\geq 3$ and prove that $n=3^{m+1}+3^m \in \mbox{ZG}(e)$. Note that $3^{m+1}\leq n < 2 \cdot 3^{m+1}$ and Corollary \ref{CorEmptyOrInfinite} is applicable. We start with some preliminaries lemmas.

\begin{lemma}
Let $n=3^{m+1}+3^{m}$ with $m\geq 1$. Then $\binom{n}{i}\not\equiv 0 \pmod{9} \Leftrightarrow i/3^{m-1} \in \{0, 1, 2, 3, 6, 9,$ $10,11, 12\}$. Moreover, in this case we have 
$$\binom{n}{i} \equiv \left\{ \begin{array}{ll} 
1\pmod{9} & \textrm{if } i/3^{m-1} \in \{0,12\};\\
3\pmod{9} & \textrm{if } i/3^{m-1} \in \{1,2,10,11\};\\
4\pmod{9} & \textrm{if } i/3^{m-1} \in \{3,9\}.\\
6\pmod{9} & \textrm{if } i/3^{m-1}=6
\end{array}  \right.$$
\end{lemma}

\begin{proof}

If $i>n$ we have $i/3^{m+1}>12$ and $\binom{n}{i}=0\equiv 0 \pmod{9}$. We assume now that $i\leq n$ and consider the base-$3$ representation of $i$ given by $i=\sum_{j=0}^{m+1}i_j 3^j$ with $i_{m+1}\in\{0,1\}$ and $i_j\in\{0,1,2\}$ if $0\leq j \leq m$. By Theorem \ref{DavisWebbTheorem} we have
\begin{equation}\label{EqBinomEta}
\binom{n}{i}\equiv \eta(i_{m+1},i_m, i_{m-1})\cdot \prod_{j=1}^{m-1} \simb{0}{i_j}^{-1}\simb{0,0}{i_j,i_{j-1}}\pmod{9},
\end{equation}
 where $\eta(i_{m+1},i_m, i_{m-1}) = \simb{1,1}{i_{m+1},i_m}\cdot \simb{1}{i_m}^{-1}\simb{1,0}{i_m, i_{m-1}}$. By direct computation we have 
$$\eta(i_{m+1},i_m, i_{m-1}) \equiv \left\{  \begin{array}{ll}
0\pmod{9} & \textrm{if }(i_{m+1},i_{m},i_{m-1})=(0,2,2), (0,2,1), (0,1,2), (0,1,1); \\
1\pmod{9} & \textrm{if }(i_{m+1},i_{m},i_{m-1})=(1,1,0), (0,0,0); \\
3\pmod{9} & \textrm{if }(i_{m+1},i_{m},i_{m-1})=(1,0,2), (1,0,1), (0,0,2), (0,0,1); \\
4\pmod{9} & \textrm{if }(i_{m+1},i_{m},i_{m-1})=(1,0,0), (0,1,0); \\
6\pmod{9} & \textrm{if }(i_{m+1},i_{m},i_{m-1})=(0,2,0).
\end{array}  \right.$$

We consider three cases:\\

\noindent If $(i_{m+1},i_{m},i_{m-1})=(0,2,2), (0,2,1), (0,1,2), (0,1,1)$ then $\binom{n}{i}\equiv 0 \pmod{9}$ (by Equation (\ref{EqBinomEta})).\\

\noindent If $(i_{m+1},i_{m},i_{m-1})=(1,0,2), (1,0,1), (0,0,2), (0,0,1)$ then $\eta(i_{m+1},i_m, i_{m-1})=3$. By Equation (\ref{EqBinomEta}) we have $\binom{n}{i} \not\equiv 0 \pmod{9}$ if and only if $\simb{0}{i_j}^{-1}\simb{0,0}{i_j,i_{j-1}}\not\equiv 0 \pmod{3}$ $\forall j: 1\leq j \leq m-1$, if and only if $i_{j-1}=0$ $\forall j: 1\leq j \leq m-1$, if and only if $i/3^{m-1}=i_{m+1}\cdot 3^{2} + i_{m}\cdot 3 + i_{m-1} \in\{11,10,2,1\}$.\\

\noindent If $(i_{m+1},i_{m},i_{m-1})=(1,1,0), (1,0,0), (0,2,0), (0,1,0), (0,0,0)$ then $\eta(i_{m+1},i_m, i_{m-1})=1,4$ or $6$. We consider two subcases. If $i_{j-1}=0$ for every $j$, $1\leq j \leq m-1$, then $\binom{n}{i}\equiv \eta(i_{m+1},i_m,i_{m-1}) \not\equiv 0 \pmod{9}$. If there exists $i_{j-1}\neq 0$ with $1\leq j \leq m-1$, we consider $j$ maximal with respect to this property. Note that $i_{m-1}=0$ and $\simb{0}{i_{j}}^{-1} \simb{0,0}{i_{j},i_{j-1}} = \simb{0,0}{0,i_{j-1}}=9$. By Equation (\ref{EqBinomEta}) we have $\binom{n}{i}\equiv 0 \pmod{9}$. Thus, $\binom{n}{i}\not\equiv 0 \pmod{9}$ if and only if $i/3^{m-1}=i_{m+1}\cdot 3^2 + i_{m}\cdot 3 + i_{m-1} \in \{0,3,6,9,12\}$.\\

By the three cases considered above we conclude that $\binom{3^{m+1}+3^m}{i}\not\equiv 0 \pmod{9}$ if and only if $i/3^{m-1}\in\{0,1,2,3,6,9,10,11,12\}$.

\end{proof}

We note that if $m\geq 2$ then $2^{3^{m-1}k}\equiv (-1)^{k} \pmod{9}$. Thus, we have the following corollary.

\begin{corollary}\label{Cormm}
If $n=3^{m+1}+3^{m}$, $m\geq 2$ then 
\begin{align}
k(n,e)= & \sum_{i=0}^{n} 2^i \binom{n}{i} \binom{e}{i} \equiv 1 - 3 \binom{e}{3^{m-1}} + 3 \binom{e}{2\cdot 3^{m-1}} -4 \binom{e}{3\cdot 3^{m-1}} + 6 \binom{e}{6\cdot 3^{m-1}} \nonumber \\ & - 4 \binom{e}{9\cdot 3^{m-1}} +3 \binom{e}{10\cdot 3^{m-1}} -3 \binom{e}{11\cdot 3^{m-1}} + \binom{e}{12\cdot 3^{m-1}}\pmod{9}. 
\end{align}
\end{corollary}

\begin{proposition}\label{PropFirstCondition}
Let $n=3^{m+1}+3^{m}$, $m\geq 2$ and $e$ be a positive integer such that $\delta_3(e)=m+1$. Then, $k(n,e)\equiv 3, 6 \pmod{9}$.
\end{proposition}

\begin{proof}
Write $e=\sum_{i=0}^{h-1}e_i  3^i$ with  $h\geq m+4$ and $e_i\in \{0,2\}$ for $m+1<i\leq h-1$, $e_{m+1}=1$ and $e_i \in \{0,1,2\}$ for $0\leq i <m+1$. We consider $k \in \{0,1,2,3,6,9,10,11,12\}$ and write $k=k_2\cdot 3^2 + k_1 \cdot 3 + k_0$ with $k_2\in\{0,1\}$ and $k_1,k_0 \in \{0,1,2\}$. By Proposition \ref{PropDavisWebbExtension} we have:
\begin{align}\label{Eqbinomial}
\binom{e}{k\cdot 3^{m-1}} \equiv & \simb{e_{h-1},e_{h-2}}{0,0} \cdot \left( \prod_{j=m+3}^{h-2}\simb{e_j}{0}^{-1}\simb{e_j,e_{j-1}}{0,0} \right) \cdot \simb{e_{m+2}}{0}^{-1}\simb{e_{m+2},e_{m+1}}{0,k_2}\cdot \simb{e_{m+1}}{k_2}^{-1}\simb{e_{m+1},e_{m}}{k_2,k_1}  \nonumber \\
&  \cdot  \simb{e_{m}}{k_1}^{-1}\simb{e_{m},e_{m-1}}{k_1,k_0} \cdot  
\simb{e_{m-1}}{k_0}^{-1}\simb{e_{m-1},e_{m-2}}{k_0,0} \cdot 
\left( \prod_{j=1}^{m-2} \simb{e_j}{0}^{-1}\simb{e_j,e_{j-1}}{0,0}\right) \nonumber \\
\equiv & \simb{e_{m+2},e_{m+1}}{0,k_2} \simb{e_{m+1},e_m}{k_2,k_1} \simb{e_m}{k_1}^{-1} \simb{e_m, e_{m-1}}{k_1,k_0} \simb{e_{m-1}}{k_0}^{-1}\simb{e_{m-1},e_{m-2}}{k_0,0} \nonumber \\ = & \simb{e_{m+2},1}{0,k_2}\cdot \simb{e_m}{k_1}^{-1} \simb{1,e_m}{k_2,k_1}\cdot \simb{e_{m-1}}{k_0}^{-1}  \simb{e_m, e_{m-1}}{k_1,k_0}\cdot \simb{e_{m-2}}{0}^{-1} \simb{e_{m-1},e_{m-2}}{k_0,0} \nonumber \\
\equiv & \binom{e_{m+2}\cdot 3^4 + 1\cdot \cdot 3^3 + e_{m}\cdot 3^2 + e_{m-1}\cdot 3 + e_{m-2}}{0\cdot 3^{4}+ k_2\cdot 3^3 + k_1\cdot 3^2 + k_0\cdot 3 + 0} \pmod{9}.
\end{align}

Let $\tilde{e} = e_{m+2}\cdot 3^4 + 3^3 + e_m\cdot 3^2 + e_{m-1}\cdot 3 + e_{m-2}$. By Corollary \ref{Cormm} and Equation \ref{Eqbinomial} we have
$$ k(n,e) \equiv  1 - 3 \binom{\tilde{e}}{3} + 3 \binom{\tilde{e}}{6} -4 \binom{\tilde{e}}{9} + 6 \binom{\tilde{e}}{18} - 4 \binom{\tilde{e}}{27} + 3 \binom{\tilde{e}}{30} - 3 \binom{\tilde{e}}{33} + \binom{\tilde{e}}{36}\pmod{9}.$$
Since $e_{m+2}\in\{0,2\}$ and $e_{m},e_{m-1},e_{m-2}\in\{0,1,2\}$ we have that $27\leq \tilde{e}\leq 53$ or $189\leq \tilde{e} \leq 215$. By direct calculation, using the above congruence formula for $k(n,e)$, we obtain
$$ k(3^{m+1}+3^{m},e) \equiv \left\{  \begin{array}{ll}
3\!\!\! \pmod{9} & \textrm{if $27\leq \tilde{e}\leq 35$ or $207\leq \tilde{e}\leq 215$};\\
6\!\!\! \pmod{9} & \textrm{if $36\leq \tilde{e}\leq 53$ or $189\leq \tilde{e}\leq 206$}. \end{array}   \right.$$
\end{proof}

Next we prove that $p(3^{m+1}+3^{m},e)\equiv 0 \pmod{3}$ if $\delta_3(e)=m+1\geq 3$. We prove first a preliminary lemma.

\begin{lemma}\label{LemmaSpecialFormulaFork}
Let $n=3^{m+1}+3^{m}$, $m\geq 2$, $h\geq m+2$ and $j<3^{h}$. If $j=\sum_{i=0}^{h-1}j_i 3^i$ with each $j_i\in\{0,1,2\}$ then $$k(n-1,j)\equiv (1-j_{m+1})\cdot (-1)^{j_0+j_1+\cdots + j_{m-1}} \pmod{3}.$$ In particular $k(n-1,j)\pmod{3}$ does not depend on $j_m$.
\end{lemma}

\begin{proof}
Since $n-1= \sum_{i=0}^{m+1} n_i 3^{i}$ with $n_{m+1}=1, n_{m}=0$ and $n_{i}=2$ for $0\leq i \leq m-1$, by Lemma \ref{LemmakIsMultiplicativeMod3} we have $k(n-1,j)\equiv k(1,j_{m+1})\cdot \prod_{i=0}^{m-1}k(2,j_i)\pmod{3}$, where $k(1,j_{m+1})=2j_{m+1}+1\equiv 1-j_{m+1}\pmod{3}$ and $k(2,j_i)=2j_i^2 + 2j_i+1\equiv (-1)^{j_i}\pmod{3}$ for $0\leq j_i \leq 2$. 
\end{proof}

\begin{lemma}\label{LemmapReduction}
Let $n=3^{m+1}+3^{m}$ with $m\geq 2$ and $e$ and $h$ be positive integers such that $\delta_3(e)=m+1$ and $h-1>m+1$. Then $p(n,2\cdot 3^{h-1}+e) \equiv p(n,e) \pmod{3}$.
\end{lemma}

\begin{proof}
We have $n-1=\sum_{i=0}^{h-1}n_i 3^i$ with $n_i=0$ for $m+1<i\leq h-1$, $n_{m+1}=1, n_{m}=0$ and $n_{i}=2$ for $0\leq i \leq m-1$. Since $n_{h-1}=0$ and $n_{h-2}\neq 2$ (because $h-2\geq m+1$), applying Lemma \ref{LemmaOtherFormulaForp} we obtain $p(n,2\cdot 3^{h-1}+e)\equiv (-1)^{n_{h-1}} p(n,e) \equiv p(n,e) \pmod{3}$.
\end{proof}

\begin{proposition}\label{PropSecondCondition}
Let $n=3^{m+1}+3^{m}$, $m\geq 2$ and $e$ be a positive integer such that $\delta_3(e)=m+1$. Then, $p(n,e)\equiv 0 \pmod{3}$.
\end{proposition}

\begin{proof}
By applying Lemma \ref{LemmapReduction} several times, it suffices to prove this proposition for the case $e<2\cdot 3^{m+1}$. In this case, by Equation (\ref{EqFormulaForpMod3}), we have $$p(n,e)\equiv 2\cdot\!\!\!\!\!\!  \doblesum{0\leq j \leq e}{j\not\equiv e\!\!\!\!\! \pmod{3}}\!\!\!\!\!  k(n-1,j) \equiv  2\cdot\!\!\!\!\!\!  \doblesum{3^{m+1}\leq j \leq e}{j\not\equiv e\!\!\!\!\! \pmod{3}}\!\!\!\!\!  k(n-1,j) + 2\cdot\!\!\!\!\!\!  \doblesum{0\leq j < 3^{m+1}}{j\not\equiv e\!\!\!\!\! \pmod{3}}\!\!\!\!\!  k(n-1,j)  \pmod{3}.$$ By Lemma \ref{LemmaSpecialFormulaFork}, since $e<2\cdot 3^{m+1}$, we have that $k(n-1,j)\equiv 0 \pmod{3}$ for $3^{m+1}\leq j \leq e$. Thus, we have $$ 2\cdot\!\!\!\!\!\!  \doblesum{3^{m+1}\leq j \leq e}{j\not\equiv e\!\!\!\!\! \pmod{3}}\!\!\!\!\!  k(n-1,j) \equiv 0 \pmod{3},$$ and by Lemma \ref{LemmaSpecialFormulaFork}
\begin{align*}
2\cdot\!\!\!\!\!\!  \doblesum{0\leq j < 3^{m+1}}{j\not\equiv e\!\!\!\!\! \pmod{3}}\!\!\!\!\!  k(n-1,j) &=  2\cdot\!\!\!\!\!\!  \doblesum{0\leq j' < 3^{m}}{j\not\equiv e\!\!\!\!\! \pmod{3}} \sum_{j_m=0}^{2}  k(n-1,j'+j_m \cdot 3^m)\\  & \equiv   
2\cdot\!\!\!\!\!\!  \doblesum{0\leq j' < 3^{m}}{j\not\equiv e\!\!\!\!\! \pmod{3}} 3\cdot k(n-1,j') \equiv 0 \pmod{3}.
\end{align*}
Therefore $p(n,e)\equiv 0 \pmod{3}$.
\end{proof}

The following classification of the Zhang-Ge sets is a consequence of Corollaries \ref{CorZempty}, \ref{CorZGnonemptyForDelta1} and \ref{CorZGnonemptyForDelta2}, and Propositions \ref{PropFirstCondition} and \ref{PropSecondCondition}.

\begin{theorem}\label{ThClassificationOfZG}
Let $e\geq 1$. Then, 
\begin{itemize}
\item $\mbox{ZG}(e)=\emptyset$ if $\delta_3(e)=0$ or $\delta_3(e)=\infty$;
\item $\mbox{ZG}(e)$ has infinitely many elements if $1\leq \delta_3(e)<\infty$.
\end{itemize}
\end{theorem}

We apply our results to the non-existence of perfect Lee codes. Propositions \ref{PropDelta1p1}, \ref{PropDelta1p2}, \ref{PropDelta2}, \ref{PropFirstCondition} and \ref{PropSecondCondition} bring us explicit elements for the Zhang-Ge set $\mbox{ZG}(e)$ (depending on $e$). If we combine these propositions together with Corollary \ref{CorEmptyOrInfinite} we obtain the following theorem.


\begin{theorem}\label{ThMain2}
If the radius $e\geq 1$ verifies $1\leq \delta_3(e)<\infty$ then $\lp{e}=\emptyset$ for infinitely many dimensions $n$. Moreover, if $\ell(e)= \lfloor \log_3(e)\rfloor +1$  we have
\begin{itemize}
\item if $\delta_3(e)=1$ and $e\equiv 3,5\pmod{9}$ then $\lp{e}=\emptyset$ for $n\equiv 12 \pmod{3^{\max\{4,\ell(e)\}}}$; 
\item if $\delta_3(e)=1$ and $e\equiv 4\pmod{9}$ then $\lp{e}=\emptyset$ for $n\equiv 3 \pmod{3^{\max\{3,\ell(e)\}}}$; 
\item if $\delta_3(e)=2$ then $\lp{e}=\emptyset$ for $n\equiv 12 \pmod{3^{\max\{4,\ell(e)\}}}$;
\item if $\delta_3(e)=m+1\geq 3$, then $\lp{e}=\emptyset$ for $n\equiv 3^{m+1}+3^{m} \pmod{3^{\max\{m+3,\ell(e)\}}}$.
\end{itemize}

\end{theorem}

The density of a subset $E\subseteq \Z^{+}$ is defined as $\mbox{dens}(E)=\lim_{N\to \infty}\frac{E\cap\{1,2,\ldots,N\}}{N}$ when this limit exists. Next we prove that the set of radii $e\geq 1$ for which we prove that $\lp{e}=\emptyset$ for infinitely many values of $n$ has density $1$.

\begin{proposition}\label{PropEhasDensity1}
The set $E=\{e\geq 1: 1\leq \delta_3(e) < \infty \}$ has density $1$.
\end{proposition}

\begin{proof}
Let $N\geq 1$ and $h\geq 1$ such that $3^{h-1}\leq N < 3^h$. We have 
\begin{align*}
1\geq & \frac{\#\{e\leq N: 1\leq \delta_3(e)<\infty\}}{N} = 1-\frac{\#\{e\leq N: \delta_3(e)=0\}}{N} -  \frac{\#\{e\leq N: \delta_3(e)=\infty\}}{N}\\ \geq & 1 -  \frac{\#\{e< 3^h: \delta_3(e)=0\}}{3^{h-1}} -  \frac{\#\{e< 3^h: \delta_3(e)=\infty\}}{3^{h-1}}= 1-\frac{2^{h-1}}{3^{h-1}} - \frac{2^{h}}{3^{h-1}}
\end{align*}
Since $h=\lfloor \log_3(N) \rfloor +1 \to \infty$ when $N\to \infty$, from the inequalities above, $\mbox{dens}(E)=1$.
\end{proof}

\section{Beyond the Zhang-Ge condition: the $Q$-polynomials}\label{SectionBeyond}

In this section we introduce a family of homogeneous polynomials and extend some results from \cite{Kim17,QCC18} and \cite{ZG17}. 

\begin{definition}
Let $B\subseteq \Z^n$ with $B$ finite and $k\in \Z^+$. The $Q$-polynomial associated with $B$ of order $k$ is $Q_B^{k}(x):= \sum_{b\in B} \langle x,b \rangle^{2k}$ where $x=(x_1,\ldots,x_n)$ and $\langle , \rangle$ denotes the standard inner product of $\R^n$. When $B=B^n(e)$ we write $Q_{(n,e)}^{k}(x)$ instead of $Q_{B^n(e)}^{k}(x)$.
\end{definition}

By definition $Q_B^{k}(x)$ is the homogeneous polynomial in $n$ variables of degree $2k$ (i.e. a $2k$-homogeneous polynomial) given by $Q_B^k (x_1,\ldots,x_n) := \sum_{b\in B} (b_1 x_1+\cdots+ b_n x_n)^{2k}$. We are interested in the case when $B=-B$, that is, when $B$ is symmetric with respect to the origin. In this case it is easy to see that $\sum_{b\in B}\langle b,x\rangle^{2k+1}=0$ for all $k\geq 0$ and for this reason we only consider even exponents. \\

We note that these polynomials have been used in very special cases to prove the non-existence of perfect Lee codes. For example the case $k=1$ was considered in \cite{ZG17} where the authors prove the non-existence of linear perfect Lee codes by the formula 
\begin{equation}\label{EquationCaseZhangGe}
Q_{(n,e)}^{1}(x) = p_1(n,e) \cdot S_2(x),
\end{equation} 
where $p_1(n,e)=  \sum_{i=0}^{e}2i^2 k(n-1,e-i)$ (see Equation (\ref{EqFormulaForp})) and $S_2(x) = \sum_{i=1}^n x_i^2$. An expression for the case $e=2$ was obtained in \cite{Kim17} to prove the non-existence of $2$-error correcting codes. This expression is given by
\begin{equation} \label{EquationCaseKim}
Q_{(n,2)}^{k}(x)= p_k(n,2) \cdot S_{2k}(x) + \sum_{t=1}^{k-1}c_t\cdot  S_{2(k-t)}(x)\cdot S_{2t}(x),
\end{equation}
where $p_k(n,2)= 4^k+4n+2$, $S_{2t}(x)= \sum_{i=1}^n x_i^{2t}$ and $c_t = 2\cdot \binom{2k}{2t}$ for $1\leq t <k$. This formula was also used in \cite{QCC18}. We note that in these papers the unique necessary information about the numbers $c_t$ to obtain the non-existence results is that they are integers. In this section we deduce a general expression for $Q_{(n,e)}^{k}(x)$ and obtain a new criterion for the non-existence of linear perfect Lee codes which generalizes Theorem \ref{ThMain1}.

\subsection{Multivariate symmetric polynomials}

In this part we review some basic results on multivariate symmetric polynomials with focus on the $Q$-polynomials. A good reference on symmetric polynomials is the book of MacDonald \cite{Macdonald98}. As usual $S_n$ denotes the set of all permutations $\theta$ of the set $[n]=\{1,\ldots,n\}$ and $R[x_1,\ldots,x_n]$ denotes the set of all polynomials in the variables $x_1,\ldots, x_n$ and coefficients in the ring $R$ (here $R=\Z$ or $R=\Q$). Let $\theta$ be a permutation of $S_n$. For $x=(x_1,\ldots,x_n)$ we denote by $\theta x$ the $n$-tuple $\theta x := (x_{\theta(1)},x_{\theta(2)},\cdots, x_{\theta(n)})$ and
for $B\subseteq \Z^n$ we denote by $\theta B$ the set $\theta B = \{\theta b: b \in B\}$. When $\theta B= B$ for all $\theta \in S_n$, we say that $B$ is \emph{$S_n$-invariant}. A polynomial $f$ in $n$ variables is called \emph{symmetric} when $f(\theta x) = f(x)$ for every $\theta \in S_n$. First we prove that the polynomials $Q_{(n,e)}^{k}(x)$ are symmetric polynomials. 

\begin{proposition}\label{PropQQisSymmetric}
Let $B \subseteq \Z^n$ be a finite and $S_n$-invariant set. The $Q$-polynomial $Q_{B}^{k}(x)$ is an homogeneous symmetric polynomial of degree $2k$. In particular, the polynomials $\QQ(x)$ are.
\end{proposition}

\begin{proof}
If $\lambda \in \R$ then $\QB(\lambda x)= \sum_{b\in B}\langle b,\lambda x \rangle^{2k} = \sum_{b\in B}\lambda^{2k}\langle b, x \rangle^{2k} = \lambda^{2k} \QB(x)$. This proves that $\QB(x)$ is an homogeneous polynomials of degree $2k$. To prove that $\QB(x)$ is a symmetric polynomial we consider $\theta \in S_n$. Since $B$ is $S_n$-invariant, the map $b\to \theta b$ establishes a bijection on $B$. Thus, 
$$ \QB(\theta x) = \sum_{b\in B} \langle b, \theta x \rangle = \sum_{b'\in \theta B} \langle b', \theta x \rangle =  \sum_{b\in B} \langle \theta b, \theta x \rangle = \sum_{b\in B} \langle b, x \rangle = \QB(x)$$ 
which proves that the polynomial $\QB(x)$ is symmetric.
\end{proof}
 
We denote by $\Lambda_{n}^{i}(R)$ the set of all $i$-homogeneous symmetric polynomials in $R[x_1,\ldots, x_n]$ where $R=\Z$ or $R=\Q$. The fundamental theorem of symmetric polynomials states that every symmetric polynomial in $ \Z[x_1,\ldots,x_n]$ can be written in a unique way as a polynomial in the elementary symmetric functions $e_1,e_2,\ldots, e_n$ (given by $e_k(x):= \sum_{1\leq i_1 < \cdots<i_k\leq n} x_{i_1}\cdots x_{i_k}$) with integer coefficients. From which it can be proved that every polynomial in $\Lambda_{n}^{i}(\Z)$ can be written as a $\Z$-linear combination of the polynomials $e_t:=e_1^{t_1}e_2^{t_2}\cdots e_n^{t_n}$ with $t=(t_1,t_2,\cdots, t_n)\in \N^n$ satisfying $t_1+2t_2+\cdots + n t_n = i$. Newton's identities express each elementary symmetric function as a polynomial in the power sum symmetric functions $S_k(x)=\sum_{i=1}^{n}x_i^k$ with rational coefficients. Consequently, each polynomial in $\Lambda_{n}^{i}(\Z)$ can be written as a $\Q$-linear combination of the power sum symmetric functions $S_\lambda:= S_{\lambda_1}S_{\lambda_2}\cdots S_{\lambda_\ell}$ where $\lambda=(\lambda_1,\lambda_2,\cdots, \lambda_{\ell})$ runs over all partitions of $i$ (i.e. $\lambda$ satisfies $\lambda_1+\lambda_2+\cdots + \lambda_{\ell} = i$ and $\lambda_1\geq \lambda_2 \geq \cdots \geq \lambda_{\ell}\geq 1$). Since $\QQ(x) \in \Lambda_{n}^{2k}(\Z)$ (Proposition \ref{PropQQisSymmetric}), there are rational numbers $p_{k}(n,e)$ and $c_\lambda$ such that
$$ \QQ(x)= p_{k}(n,e)\cdot S_{2k} + \sum_{\lambda} c_\lambda \cdot S_{\lambda}, $$
where $\lambda=(\lambda_1,\ldots, \lambda_{\ell})$ runs over all partitions of $2k$ with length $\ell>1$. Our first goal is to prove that the numbers $p_{k}(n,e)$ and $c_\lambda$ are integers, to find an explicit expression for $p_{k}(n,e)$ and to prove that $S_{\lambda}=0$ when some coordinate of $\lambda$ is odd.


\subsection{Explicit formulas for the $Q$-polynomials}

Here we deduce some explicit formulas for the $Q$-polynomials. Let $i=(i_1,\ldots, i_n)\in \N^n$ and $k=i_1+\cdots+i_n$. We denote the corresponding multinomial coefficient by $\binom{k}{i}=\frac{k!}{i_1!\cdots i_n!}$ and by $x^i$ the monomial $x^i = x_1^{i_1}\cdots x_n^{i_n}$. In order to indicate that $j=(j_1,\ldots, j_r)$ is a partition of $k$ we use, as usual, the notation $j \vdash k$. The length of $j$ will be denoted by $\ell(j)=r$. Let $\mathcal{P}(k,s)$ denote the set of all partitions $j\vdash k$ with $\ell(j)\leq s$. The minimum of two integers $a$ and $b$ will be denoted by $a \wedge b := \min\{a,b\}$. We say that two $n$-tuples $x,y \in \N^n$ are $S_n$-equivalent when $x=\theta y$ for some $\theta \in S_n$ and denote it by $x\sim y$. It is easy to see that the $S_n$-equivalence is an equivalence relation. Before deducing a formula for the $Q$-polynomials we need some preliminary lemmas.

\begin{lemma}\label{Lemma2forQformula}
Let $i,j \in \N^n$, $k=i_1+\cdots+i_n$ and $B\subseteq \Z^n$. If $B$ is an $S_n$-invariant subset and $i\sim j$, then $\binom{2k}{2i}=\binom{2k}{2j}$ and $\sum_{b\in B}b^{2i} = \sum_{b\in B}b^{2j}$.
\end{lemma} 

\begin{proof}
The first equality is clear. To prove the second equality we consider $\theta \in S_n$ such that $i=\theta j$. Since $B$ is $S_n$-invariant, the map $b \to \theta b$ induces a bijection on $B$. Thus, $$\sum_{b\in B}b^{2i}=\sum_{b\in B}(\theta b)^{2i}= \sum_{b\in B}(\theta b)^{2\cdot\theta j}= \sum_{b\in B}b_{\theta(1)}^{2j_{\theta(1)}} \cdots   b_{\theta(n)}^{2j_{\theta(n)}} = \sum_{b\in B} b_1^{2j_1}\cdots b_n^{2j_n}= \sum_{b\in B} b^{2j}.$$
\end{proof}

\begin{definition}
If $B \subseteq \Z^n$ and $j \in \N^n$ we denote by $B^{(j)}= \sum_{b\in B} b^j$.
\end{definition}

\begin{lemma}\label{Lemma1forQformula}
Let $B\subseteq \Z^n$ such that $B=-B$. If $i=(i_1,\ldots,i_n)\in \N^n$ satisfies $\sum_{b\in B}b^{i}\neq 0$ then $i_s$ is even for every $s$, $1\leq s \leq n$.
\end{lemma}

\begin{proof}
By contradiction, we suppose that $i_s$ is odd for some $s \in \{1,\ldots, n\}$ and consider the map $\phi_s:B \rightarrow B$ which changes the sign of the $s$-th coordinate. We have that $\phi_s$ is a bijection (because $B=-B$) and $\phi_s(b)^{i}=- b^i$ (because $i_s$ is odd). Thus $\sum_{b\in B}b^{i} = \sum_{b\in B}\phi_s(b)^{i} = - \sum_{b\in B}b^{i} $ and then $\sum_{b\in B}b^{i} = 0$.
\end{proof}

\begin{definition}\label{Def-e-regular}
A subset $B \subseteq \Z^n$ is regular if it is $S_n$-invariant and $B=-B$. If in addition, every $b\in B$ has at most $e$ non-zero coordinates, for some positive integer $e$, we say that $B$ is $e$-regular.
\end{definition}

Note that every $\ell_p$-ball ($1\leq p \leq \infty$) given by $B_{p}^{n}(e)=\{x\in \Z^n: \sqrt[p]{|x_1|^p + \cdots + |x_n|^p}\leq e\}$ is a regular set and the Lee ball $B^n(e)$ is $e$-regular. Next we derive a formula for the $Q$-polynomial associated with a regular set $B$ as a linear combination of symmetric monomials. 


\begin{remark}\label{RemarkExponentiation}
Let $k\in \Z^+$, $j$ be a partition of $k$ and $x=(x_1,\ldots,x_n)$ be an $n$-tuple. From here on we use the standard convention of defining $x^j=0$ when $\ell(j)>n$ and $x^j = x^{j^*}$ when $\ell(j)<n$, where $j^*$ is the $n$-tuple which coincides with $j$ in the first $\ell(j)$ coordinates and is zero in the remaining $n-\ell(j)$ coordinates.
\end{remark}


\begin{proposition}\label{PropFormulaForQB-monomial}
Let $B\subseteq \Z^n$ be a regular set. Then,
$$Q_{B}^{k}(x) = \sum_{j\in \mathcal{P}(k,n)} \binom{2k}{2j}\left(\sum_{b\in B}b^{2j} \right) m_{2j}(x),   $$ where $$m_{2j}(x)=\sum_{s \sim j}x^{2s}$$ is the monomial symmetric function associated with the partition $2j\vdash 2k$. Moreover, if $B$ is $e$-regular the sum corresponding to $\QB(x)$ can be restricted to the partitions $j \in \mathcal{P}(k,n \wedge e)$.
\end{proposition}

\begin{proof}
Applying the multinomial theorem to $\langle b,x \rangle^{2k}=(b_1x_1+\ldots+b_nx_n)^{2k}$ we obtain: 
$$ (b_1x_1+\ldots+b_nx_n)^{2k} = \sum_{i_1+\cdots+i_n=2k}\frac{(2k)!}{i_1! \cdots i_n!}\cdot b_1^{i_1}\cdots b_n^{i_n}\cdot x_1^{i_1}\cdots x_n^{i_n} = \sum_{i_1+\cdots+i_n=2k} \binom{2k}{i}b^{i}x^{i}.$$
Substituting the above value in $\QB(x) = \sum_{b\in B} \langle b,x \rangle^{2k}$ we obtain
\begin{equation}\label{Equation1forQ}
 \QB(x) = \sum_{b\in B}\left( \sum_{i_1+\cdots+i_n=2k} \binom{2k}{i}b^{i}x^{i}\right) =\sum_{i_1+\cdots+i_n=2k} \binom{2k}{i}\left(\sum_{b\in B}   b^{i} \right)  x^{i}. 
\end{equation}

Since $B=-B$, by Lemma \ref{Lemma1forQformula}, every non-zero term in the last sum of Equation (\ref{Equation1forQ}) corresponds to even values of $i$. For these values, we can write $i=2j$ where $j=(j_1,\ldots, j_n)\in \N^n$ satisfying $j_1+\ldots+j_n=k$. Then, 
\begin{equation}\label{Equation2forQ}
\QB(x) = \sum_{j_1+\ldots+j_n=k} \binom{2k}{2j}\left( \sum_{b\in B}  b^{2j} \right) x^{2j} = \doblesum{j_1+\ldots+j_n=k}{j_1\geq j_2 \geq \cdots \geq j_n} \left( \sum_{h\sim j} \binom{2k}{2h}\left( \sum_{b\in B}  b^{2h} \right) x^{2h} \right). 
\end{equation}
Since $B$ is $S_n$-invariant, by Lemma \ref{Lemma2forQformula} and Equation (\ref{Equation2forQ}) (see also Remark \ref{RemarkExponentiation}) we have:
\begin{align}\label{Equation3forQ}
\QB(x) &= \sum_{j \in \mathcal{P}(k,n)} \left( \sum_{h\sim j} \binom{2k}{2j}\left( \sum_{b\in B}  b^{2j} \right) x^{2h} \right) = \sum_{j \in \mathcal{P}(k,n)}  \binom{2k}{2j}\left( \sum_{b\in B}  b^{2j}\right) \left(\sum_{h\sim j}  x^{2h} \right)\nonumber \\ &= \sum_{j \in \mathcal{P}(k,n)}  \binom{2k}{2j}\left( \sum_{b\in B}  b^{2j}\right) m_{2j}(x).
\end{align}
If $B$ is $e$-regular, every $b\in B$ has at most $e$ non-zero coordinates and we have $\sum_{b\in B}  b^{2j}=0$ when $\ell(j)>e$. Then, we can rewrite Equation (\ref{Equation3forQ}) as 
\begin{equation}\label{EquationForQmonomial}
 \QB(x) = \sum_{j \in \mathcal{P}(k,n \wedge e)}  \binom{2k}{2j}\left( \sum_{b\in B}  b^{2j}\right) m_{2j}(x)
\end{equation}
\end{proof}

The {\it augmented monomial symmetric functions} are defined by $\tilde{m}_{j}(x) := t_1! t_2! \cdots t_{j_1}!\cdot  m_{j}(x)$, where $j=(j_1,\ldots, j_{\ell})$ is a partition of some positive integer $k$ and $t_h=\#\{i: j_i= h \}$ for $1\leq h \leq j_1$. In \cite{Merca15}, M. Merca obtain a nice expression for the expansion of augmented monomial symmetric functions into power sum symmetric functions which we state below. A partition of $[\ell]=\{1,\ldots,\ell\}$ is a set of the form $\nu=\{\nu_1,\ldots, \nu_r\}$, where the $\nu_i$ are non-empty disjoint sets (for $1\leq i \leq r$) whose union is $[\ell]$. We denote by $\mathcal{P}_{\ell}$ the set of all partitions of $[\ell]$. For $j=(j_1,\ldots, j_{\ell}) \vdash k$ and $\nu \in \mathcal{P}_{\ell}$ the symbol $j * \nu$ is used to denote the new partition of $k$ whose parts are given by $\sum_{j \in \nu_i}j_{j}$, $1\leq i \leq |\nu|$.

\begin{lemma}[{\cite[Theorem 2]{Merca15}}]\label{LemmaMerca}
Let $k$ be a positive integer and $j = (j_1,\ldots, j_{\ell})$ be a partition of $k$. Then 
\begin{equation}
 \tilde{m}_{j}(x) = \sum_{\nu \in \mathcal{P}_{\ell}} \mu(\nu) \cdot S_{j * \nu}(x),   
\end{equation}
where $\mu(\nu)= \prod_{i=1}^{|\nu|} (-1)^{|\nu_i|-1}(|\nu_i|-1)!$ and $S_{t}(x)= \prod_{i=1}^{r}\left(\sum_{j=1}^{n}x_j^{t_i} \right)$ denotes the power sum symmetric functions associated with the partition $t=(t_1,\ldots,t_r)\vdash k$.
\end{lemma}

Let $j =(j_1,\ldots, j_\ell)\vdash k$ and $t_h$ denote the quantity $t_h=\#\{i: j_i=h\}$ for $1\leq h \leq j_1$. The reduced multinomial coefficient $\binom{k}{j}'$ is defined as $\binom{k}{j}' = \frac{1}{t_1!t_2! \cdots t_{j_1}!} \cdot \binom{k}{j}$. This coefficient matches the number of set paritions $\nu=\{\nu_1, \ldots, \nu_t \} \in \mathcal{P}_{k}$ such that $t=t_1+\ldots+t_{j_1}$ and $\#\{i: |\nu_i|=h\}=t_h$ for $1\leq h \leq j_1$ (see \cite{AS64}, pp.~823). In particular the reduced multinomial coefficients are positive integers\footnote{In \cite{AS64} the reduced multinomial coefficient is denoted by $(k; t_1,\ldots, t_{j_1})'$ instead of $\binom{k}{j}'$, where $j=(j_1,\ldots, j_{\ell})$ and $t_h=\#\{i: j_i=h\}$ for $1\leq h \leq j_1$.}. 

\begin{proposition}\label{PropFormulaFinalForQB}
Let $B \subseteq \Z^n$ be a regular set. Then, 
\begin{equation}\label{EqForQB}
\QB(x) = p_k(B)\cdot S_{2k}(x) +  \sum_{j \in \mathcal{P}'(k)} c_{j}(B,k) \cdot S_{2j}(x),
\end{equation} 
where $\mathcal{P}'(k)=\{j: j \vdash k, j \neq (k)\}$, $c_j(B,k)$ is an integer number for every $j \in \mathcal{P}'(k)$, and $p_k(B)$ is given by 
\begin{equation}\label{EqMainTermforQB}
p_k(B) = \sum_{j \in \mathcal{P}(k, n)} \binom{2k}{2j}' \cdot \left(  \sum_{b \in B} b^{2j} \right) \cdot (-1)^{\ell(j)-1} \cdot (\ell(j)-1)!. 
\end{equation}
Moreover, if $B$ is $e$-regular we have
\begin{equation}\label{EqMainTermforQBred}
p_k(B) = \sum_{j \in \mathcal{P}(k, n \wedge e)} \binom{2k}{2j}' \cdot \left(  \sum_{b \in B} b^{2j} \right) \cdot (-1)^{\ell(j)-1} \cdot (\ell(j)-1)!.
\end{equation}
\end{proposition}

\begin{proof}

We use Equation (\ref{Equation3forQ}) together with the relation $\binom{2k}{2j} m_{2j}(x) = \binom{2k}{2j}' \tilde{m}_{2j}(x)$ to express $\QB(x)$ in terms of the augmented monomial symmetric functions. Then, we use Lemma \ref{LemmaMerca} to express it in terms of the power sum symmetric functions as follows:
\begin{equation}\label{EqForQBcomplete}
\QB(x)= \sum_{j \in \mathcal{P}(k,n)}\left( \sum_{\nu \in \mathcal{P}_{\ell(j)}}  \binom{2k}{2j}' \cdot \left( \sum_{b\in B}  b^{2j}\right)\cdot  \mu(\nu) \cdot S_{2j * \nu}(x)\right).
\end{equation} Developing Equation (\ref{EqForQBcomplete}) we obtain the $Q$-polynomial $\QB(x)$ written as a $\Z$-linear combination of power sum symmetric functions of the form $S_{2j}$ with $j \vdash k$. In order to obtain an expression for the main coefficient (i.e. the coefficient corresponding to the partition $j=(k)$), we note that $2j*\nu = (2k)$ if and only if $\nu =\{\{1,2,\ldots, \ell(j)\}\}$. Thus, we obtain Equation (\ref{EqMainTermforQB}) by noting that the only term which contributes to the main coefficient in the inner sum of Equation (\ref{EqForQBcomplete}) is the corresponding to $\nu = \{\{1,2,\ldots, \ell(j)  \}\}$ and for this value of $\nu$ we have $\mu(\nu)= (-1)^{\ell(j)-1}(\ell(j)-1)!$. If $B$ is $e$-regular we proceed in a similar way but using Equation (\ref{EquationForQmonomial}) instead of Equation (\ref{Equation3forQ}) and we obtain Equation (\ref{EqMainTermforQBred}).
\end{proof}

Now we consider the case $B=B^n(e)$. In this case $p_k(B)$ will be denoted by $p_k(n,e)$. The following proposition provides a more explicit formula to compute $p_k(n,e)$.

\begin{proposition}
Let $k(n,e)=\sum_{i=0}^{n \wedge e} 2^i \binom{n}{i} \binom{e}{i}$ (with the convention that $\binom{a}{b}=0$ when $a<0$) and $p_k(n,e)$ be the main coefficient\footnote{That is, the coefficient corresponding to $S_{2k}(x)$.} of the $Q$-polynomial $\QQ(x)$. Then:
\begin{equation}\label{EqForpkne}
p_k(n,e) = \sum_{\ell=1}^{n \wedge e} \doblesum{j \vdash k}{\ell(j)=\ell} \binom{2k}{2j}' \cdot  (-1)^{\ell-1} \cdot (\ell-1)!\cdot \left(  \sum_{i_1+\cdots + i_{\ell+1}=e} 2^{\ell} i_{1}^{2j_1} \cdots i_{\ell}^{2j_{\ell}} k(n-\ell,i_{\ell+1}) \right).
\end{equation}
\end{proposition}

\begin{proof}
By Equation (\ref{EqMainTermforQBred}), it suffices to prove that 
\begin{equation}\label{EqItsuffices}
\sum_{b \in B^n(e)} b^{2j} =  \sum_{i_1+\cdots + i_{\ell+1}=e} 2^{\ell}\cdot i_{1}^{2j_1} \cdots i_{\ell}^{2j_{\ell}}\cdot k(n-\ell,i_{\ell+1}),
\end{equation}
for every $j=(j_1,\ldots,j_\ell) \vdash k$ with $\ell\leq n \wedge e$. Let $b=(b_1,\ldots,b_ n) \in B^n(e)$. We define $i_s = |b_s|$ for $1\leq s \leq \ell$ and $b'=(b_{\ell+1}, \ldots, b_{n})$. We have that $|b|\leq e$ if and only if $|b'|+i_1+\ldots+i_{\ell}\leq e$. Then,
\begin{align*}
 \sum_{b \in B^n(e)} b^{2j} &= \sum_{i_1+\cdots + i_{\ell}\leq e}\left( \sum_{b' \in B^{n-\ell}\left(e-\sum_{s=1}^{\ell}i_s \right)}  (\pm i_1)^{2j_1}\cdots (\pm i_\ell)^{2j_{\ell}} \right)\\   &=  \sum_{i_1+\cdots + i_{\ell}\leq e} 2^{\ell} \cdot i_{1}^{2j_1} \cdots i_{\ell}^{2j_{\ell}}\left( \sum_{b' \in B^{n-\ell}\left(e-\sum_{s=1}^{\ell}i_s \right)} 1\right) \\ &=  \sum_{i_1+\cdots + i_{\ell+1}=e}2^{\ell}\cdot i_{1}^{2j_1} \cdots i_{\ell}^{2j_{\ell}}\cdot k(n-\ell,i_{\ell+1}).
\end{align*}
\end{proof}

\begin{remark}
Equation (\ref{EqItsuffices}) has a nice interpretation in terms of generating functions. In \cite{Post75}, Post consider the generating function $S_n(x) = \sum_{i=0}^{\infty} k(n,i)x^{i} $ which is given by $S_n(x)= \frac{(1+x)^n}{(1-x)^{n+1}}$ if $n\geq 0$ and $S_n(x)=0$ if $n<0$. We consider here, the generating function $F_j(x)$ given by $F_j(x)= \sum_{i=0}^{\infty} i^{j} x^i$ if $j\in \Z^{+}$ and $F_j(x)= F_{j_1}(x)\cdots F_{j_{\ell}}(x)$ if $j=(j_1,\ldots, j_\ell)$. By the convolution formula, we have that $2^{-\ell}\cdot \sum_{b \in B^n(e)} b^{2j}$ is the coefficient of $x^e$ of the product $F_{2j}(x)\cdot S_{n-\ell}(x)$.
\end{remark}

We note that when $\ell>n$ we have $k(n-\ell,c)=0$ and when $\ell>e$ and $i_1,\ldots,i_{\ell+1}$ are natural numbers such that $i_1+\cdots+i_{\ell+1}=e$, we have that $i_1^{2j_1}\cdots i_\ell^{2j_\ell}k(n-\ell,i_{\ell+1})=0$. Thus, it is possible to write Equation (\ref{EqForpkne}) as
\begin{equation}\label{EqForpkneExt}
p_k(n,e) =  \sum_{j \vdash k} \binom{2k}{2j}' \cdot  (-1)^{\ell-1}  (\ell-1)! \left(  \sum_{i_1+\cdots + i_{\ell+1}=e} 2^{\ell} i_{1}^{2j_1} \cdots i_{\ell}^{2j_{\ell}} k\left(n-\ell,i_{\ell+1}\right) \right),
\end{equation}
where $\ell=\ell(j)$ (the length of the partition $j$). In order to avoid considering several cases, the above expression is convenient when we are looking for an explicit expression for $p_k(n,e)$ and a given value of $k$.

\begin{example}
For $k=1$, the only partition of $k$ is $j=(1)$ whose length is $\ell=1$. Then, Equation (\ref{EqForpkneExt}) reduces to 
\begin{equation}\label{EqForp1ne}
p_1(n,e)=\sum_{i_1+i_2=e}2\cdot i_1^{2}\cdot k(n-1,i_2)= \sum_{i=0}^{e}2i^2 k(n-1,e-i),
\end{equation}
which coincides with the expression considered by Zhang and Ge in Equation (\ref{EquationCaseZhangGe}).
\end{example}

\begin{example}
For $k=2$, the only partition of length $\ell=1$ is $j=(2)$. The corresponding term is given by $\sum_{i_1+i_2=e}2\cdot i_1^{4}\cdot k(n-1,i_2)$. The only partition of length $\ell=2$ is $j=(1,1)$. Since $\binom{4}{2,2}'\cdot  (-1)^1\cdot 1! = -3$, the corresponding term is given by $-3\cdot \sum_{i_1+i_2+i_3=e} 4i_1^2 i_2^2 k(n-2,i_3)$. Then,
\begin{equation}\label{EqForp2ne}
p_2(n,e) = 2\cdot\sum_{a+b=e}a^4 k(n-1,b) - 12\cdot \sum_{a+b+c=e}a^2 b^2 k(n-2,c)
\end{equation}
\end{example}

\begin{example}
For $k=3$, the only partition of length $\ell=1$ is $j=(3)$. The corresponding term is $\sum_{i_1+i_2=e}2 i_1^{6} k(n-1,i_2)$. There is only one partition of $3$ of length $\ell=2$ which is given by $j=(2,1)$. Since $\binom{6}{4,2}'\cdot  (-1)^1\cdot 1! = -15$, the corresponding term is $-15 \sum_{i_1+i_2+i_3=e}4i_1^4 i_2^2 k(n-2, i_3)$. There is also only one partition of $3$ of length $\ell=3$ which is given by $j=(1,1,1)$. Since $\binom{6}{2,2,2}' \cdot (-1)^2 \cdot 2! = 30$, the corresponding term is $30\sum_{i_1+\cdots+i_4=e}8i_1^2 i_2^2 i_3^2 k(n-3,i_4)$. Then,
\begin{equation}\label{EqForp3ne}
p_3(n,3) = 2\sum_{a+b=e}a^6 k(n-1,b) - 60 \sum_{a+b+c=e}a^4 b^2 k(n-2,c) + 240 \sum_{a+b+c+d=e}(abc)^2 k(n-3,d)
\end{equation}
\end{example}

\subsection{A criterion for the non-existence of perfect Lee codes}

In this part we deduce a general criterion for the non-existence of certain lattice tilings (depending on a prime number $p$). Then, we apply this criterion to the non-existence of linear perfect Lee codes. 

\begin{theorem}\label{TheoremMainGeneralized}
Let $p>2$ be a prime. If $B \subseteq \Z^n$ is a regular set such that $|B|=pm$ with $p\nmid m$ and the following congruences are satisfied:
\begin{equation}\label{EqHypothesis}
\left\{ \begin{array}{ll}
p_k(B) \not\equiv 0 \pmod{p} & \textrm{for }1\leq k < \frac{p-1}{2}\\
p_k(B) \equiv 0 \pmod{p} & \textrm{for }k=\frac{p-1}{2}
\end{array}   \right.
\end{equation}
then there is no lattice tiling of $\Z^n$ by $B$.
\end{theorem} 

\begin{proof}
By contradiction, suppose that there is a lattice tiling of $\Z^n$ by $B$. Then, by Theorem \ref{ThHorakCriterion}, there is an abelian group $G$ and an homomorphism $\phi: \Z^n \rightarrow G$ such that the restriction $\phi|_B : B \rightarrow G$ is a bijection. By Lemma \ref{LemmaEasy}, there is a surjective homomorphism $\phi': G \rightarrow \Z_p$. Then, the restriction of the homomorphism $\psi:= \phi' \circ \phi : \Z^n \rightarrow \Z_p$ to $B$ is an $m$-to-$1$ map. Let $\xi$ be a primitive root modulo $p$. We have the following congruences modulo $p$:
\begin{align}\label{EquationAbove}
\sum_{b\in B}\psi(b)^{2k} &\equiv m\cdot \left( \sum_{i=1}^{p-1}i^{2k}\right) \equiv m \cdot \left( \sum_{j=0}^{p-2}(\xi^j)^{2k}\right)  \nonumber \\ &\equiv  m \cdot \left( \sum_{j=0}^{p-2}(\xi^{2k})^{j}\right) \equiv \left\{\begin{array}{ll}
\frac{1-\xi^{2k(p-1)}}{1-\xi^{2k}}\equiv 0  & \textrm{if }0<2k<p-1;\\
m(p-1)\equiv -m  & \textrm{if } 2k=p-1.
\end{array}   \right.
\end{align}
Let $\{e_1,\ldots,e_n\}$ be the standard basis of $\R^n$. We consider the $n$-tuple $x=(x_1,\ldots,x_n)\in \Z^n$ such that $x_i \equiv \psi(e_i) \pmod{p}$ and $0\leq x_i <p$ for $1\leq i \leq n$. We have that $\QB(x)=\sum_{b\in B} \left(\sum_{i=1}^n b_i x_i\right)^{2k} = \sum_{b\in B}\psi^{2k}\left( \sum_{i=1}^n b_ie_i \right) = \sum_{b\in B}\psi^{2k}(b)$. Thus, by Proposition \ref{PropFormulaFinalForQB} and Equation (\ref{EquationAbove}) we have that
\begin{equation}\label{EqA}
p_k(B)\cdot S_{2k}(x) +  \sum_{j \in \mathcal{P}'(k)} c_{j}(B,k) \cdot S_{2j}(x) \equiv 0 \pmod{p} 
\end{equation}
for $1\leq k < \frac{p-1}{2}$ and
\begin{equation}\label{EqB}
p_{\frac{p-1}{2}}(B)\cdot S_{p-1}(x) +  \sum_{j \in \mathcal{P}'\left(\frac{p-1}{2}\right)} c_{j}\left(B,\frac{p-1}{2}\right) \cdot S_{2j}(x) \not\equiv 0 \pmod{p}.
\end{equation}
Using Equations (\ref{EqHypothesis}) and (\ref{EqA}), and the fact that if $j=(j_1,\ldots,j_{\ell}) \in \mathcal{P}'(k)$ then $S_{2j_1}(x)\mid S_{2j}$ with $j_1<k$ (because $\ell>1$), it is easy to prove by induction that $S_{2k}\equiv 0 \pmod{p}$ for $1\leq k < \frac{p-1}{2}$. This implies that $S_{2j}(x)\equiv 0 \pmod{p}$ for every $j\in\mathcal{P}'\left(\frac{p-1}{2}\right)$ and Equation (\ref{EqB}) becomes
$$ p_{\frac{p-1}{2}}(B)\cdot S_{p-1}(x) \not\equiv 0 \pmod{p},  $$
which is a contradiction because by hypothesis we have that $p_{\frac{p-1}{2}}(B)\equiv 0 \pmod{p}$.
\end{proof}

We are mainly interested in Lee codes, however Theorem \ref{TheoremMainGeneralized} can be applied also to codes with respect to the $\ell_p$ metric ($1\leq p \leq \infty$) since the balls for these metrics are regular sets. In the following example we prove the non-existence of certain $\ell_2$-codes.

\begin{example}\label{Example-Euclidean}
There are no linear perfect $8$-error-correcting codes in $\Z^3$ with respect to the Euclidean ($\ell_2$) metric. Indeed, the  ball $B= \{b\in \Z^3: \sqrt{b_1^2+b_2^2+b_3^2}\leq 8 \}$ has cardinality $|B|=2109$ and we can apply Theorem \ref{TheoremMainGeneralized} with $p=3$ (and $m=703$). Equation (\ref{EqHypothesis}) is equivalent to $p_1(B)\equiv 0 \pmod{3}$. By direct calculation we have $p_1(B)=\sum_{b\in B}b_1^2 = 26688\equiv 0 \pmod{3}$.
\end{example}

Let $p=2q+1$ be a prime (with $q\in \Z^+$). We say that a pair of positive integers $(n,e)$ satisfies the \emph{$p$-condition of non-existence} if it verifies the following system 
\begin{equation}
\left\{ \begin{array}{l}
k(n,e) \equiv tp \pmod{p^2} \textrm{ for some }t: 1 \leq t<p; \\
p_{i}(n,e) \not\equiv 0 \pmod{p} \textrm{ for every }i: 1\leq i < \frac{p-1}{2}; \\
p_{\frac{p-1}{2}}(n,e) \equiv 0 \pmod{p}.
\end{array}    \right.
\end{equation}

The following proposition is a direct corollary of Theorem \ref{TheoremMainGeneralized}.

\begin{proposition}\label{Prop-Specialization-for-Lee}
If $(n,e)$ verifies the $p$-condition of non-existence for some prime $p>2$, then $\lp{e}=\emptyset$.
\end{proposition}

In Section \ref{SectionMain} we proved that if $e$ satisfies $1\leq \delta_3(e)<\infty$, there are infinitely many dimensions $n$ such that $\lp{e}=\emptyset$ (Theorem \ref{ThMain2}). This result was obtained using the Zhang-Ge condition (Equation \ref{EqZhangGesystem}) which is equivalent to the $3$-condition of non-existence. It is possible to extend this result to other values of $e$ considering the $p$-condition of non-existence for other primes $p\neq 3$. For example, using the $5$-condition of non-existence we can extend the above result to the case $e=2$, see \cite[Theorem 1]{Qureshi18}. For this case, since the formulas for $p_1(n,e)$ and $p_2(n,e)$ can be obtained directly from the Kim's formula (Equation (\ref{EquationCaseKim})), it is not necessary to use the full potential of Equation (\ref{EqForpkneExt}). To finish this section, we use the $5$-condition of non-existence together Equation (\ref{EqForpkneExt}) to prove that there are infinitely many dimensions $n$ such that $\lp{e}=\emptyset$ for $e=6$ and $7$.

\begin{proposition}\label{Prope6}
If $n\equiv 22,47,72,97$ or $122 \pmod{125}$, then $\lp{6}=\emptyset$.
\end{proposition}

\begin{proof}
By direct calculation using Equations (\ref{EqForkne1}), (\ref{EqForp1ne}) and (\ref{EqForp2ne}) we obtain:\\
$k(n,6)=(4n^6+12n^5+70n^4+120n^3+196n^2+138n+45)/45$,\\
$p_1(n,6)=(8n^5+60n^4+280n^3+720n^2+1032n+630)/15$,\\
$p_2(n,6)= (8n^5+180n^4+1480n^3+6360n^2+14232n+13230)/15$.\\
Since $N\equiv 5,10,15\textrm{ or }20\pmod{25} \Leftrightarrow 45N\equiv 25,50,75\textrm{ or }100 \pmod{125}$ and $N\equiv 0 \pmod{5} \Leftrightarrow 15N\equiv 0 \pmod{25}$, the pair $(n,6)$ satisfies the $5$-condition of non-existence if and only if it verifies the following system of congruences:
$$ \left\{ \begin{array}{l}
4n^6+12n^5+70n^4+120n^3+196n^2+138n+45 \equiv 25,50,75\textrm{ or }100 \pmod{125}, \\
8n^5+60n^4+280n^3+720n^2+1032n+630 \not\equiv 0 \pmod{25},\\
8n^5+180n^4+1480n^3+6360n^2+14232n+13230  \equiv 0 \pmod{25}.
\end{array}  \right.  $$
Since every integer solution of this system is also a solution modulo $125$, it suffices to check the possible solutions with $0\leq n <125$. Then, the solutions are the positive integers $n$ such that $n\equiv 22,47,72,97\textrm{ or }122 \pmod{125}$.
\end{proof}

\begin{proposition}\label{Prope7}
If $n\equiv 13,23,38,48,63,73,88,98,113\textrm{ or }123 \pmod{125}$, then $\lp{7}=\emptyset$.
\end{proposition}

\begin{proof}
By direct calculation using Equations (\ref{EqForkne1}), (\ref{EqForp1ne}) and (\ref{EqForp2ne}) we obtain:\\
$k(n,7)= (8n^7+28n^6+224n^5+490n^4+1232n^3+1372n^2+1056n+315)/315   $,\\
$p_1(n,7)=(8n^6+72n^5+440n^4+1560n^3+3512n^2+4488n+2520)/45   $,\\
$p_2(n,7)= (8n^6+216n^5+2240n^4+13080n^3+44912n^2+85704n+70560)/45      $.\\
Since $N\equiv 5,10,15\textrm{ or }20\pmod{25} \Leftrightarrow 315N\equiv 25,50,75\textrm{ or }100 \pmod{125}$ and $N\equiv 0 \pmod{5} \Leftrightarrow 45N\equiv 0 \pmod{25}$, the pair $(n,7)$ satisfies the $5$-condition of non-existence if and only if it verifies the following system of congruences:
$$\left\{ \begin{array}{l}
8n^7+28n^6+224n^5+490n^4+1232n^3+1372n^2+1056n+315 \equiv 25,50,75,100 \!\!\!\!\pmod{125}\\
8n^6+72n^5+440n^4+1560n^3+3512n^2+4488n+2520\not\equiv 0\!\!\!\! \pmod{25} \\
8n^6+216n^5+2240n^4+13080n^3+44912n^2+85704n+70560 \equiv 0\!\!\!\! \pmod{25}
\end{array} \right.$$
As in the proof of Proposition \ref{Prope6}, we can restrict the possible values of $n$ to $0\leq n < 125$. Then, the solutions of the above system are the positive integers $n$ such that $n\equiv 13$, $23$, $38$, $48$, $63$, $73$, $88$, $98$, $113$ or $123 \pmod{125}$.
\end{proof}

\bibliographystyle{alpha}

\section*{Supplementary appendix}

In this appendix we present the intermediate computation used to prove some results of this paper (specifically in Propositions 5, 6, 7, 8, 16 and 17, Lemma 7 and Example 5). The software used to implement the algorithms is SAGE \cite{Sage}.

\subsection*{Intermediate computation in Section III}

\paragraph{Computation for Proposition 5}

First we calculate the integer coefficient polynomial $f(e)=12! \cdot k(12,e)=  12! \sum_{i=0}^{12} 2^{i} \binom{12}{i} \binom{e}{i}$ with the following code:

\includegraphics[width=1.0\textwidth]{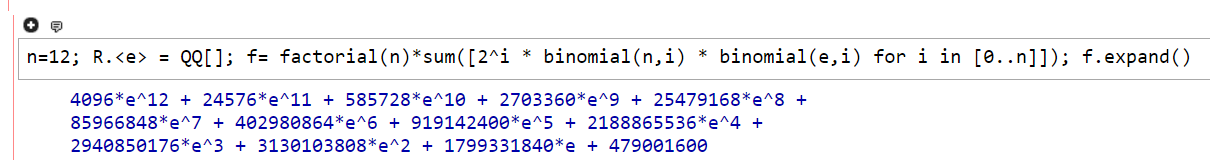}

Next, we check that for the $486$ elements of the form $e=a+3^2b$ with $a\in\{3,5\}$ and $0\leq b < 3^5$, we have $f(e)\equiv 3^{6}$ or $2\cdot 3^6 \pmod{3^7}$. We consider the list $L$ consisting of all the elements $e$ verifying this property and check that the length of $L$ is exactly $486$.

\includegraphics[width=1.0\textwidth]{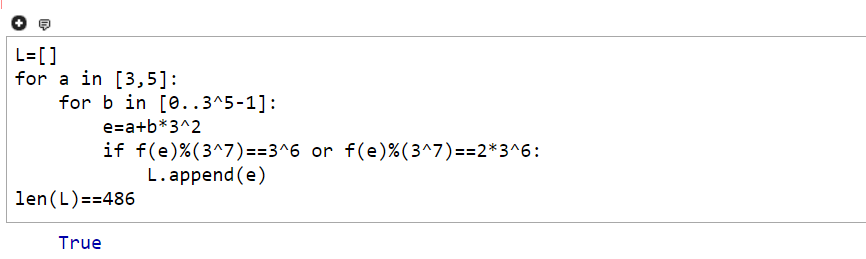}

Finally, we compute $p(12,3)$ and $p(12,5)$ and check that both values are multiple of $3$.

\includegraphics[width=1.0\textwidth]{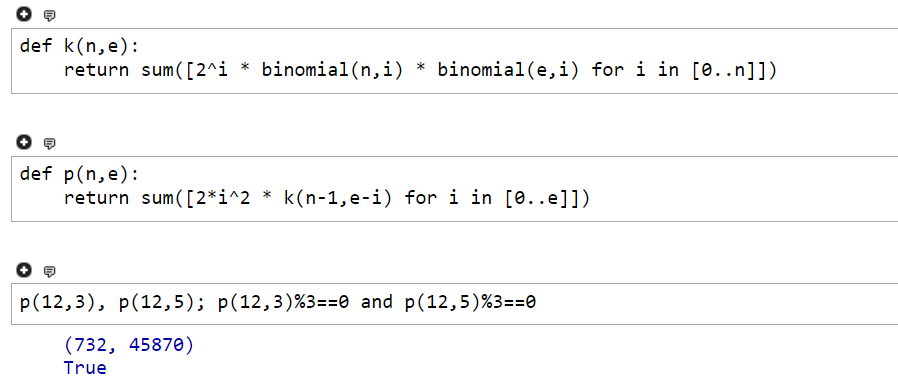}
\vspace{2mm}

\paragraph{Computation for Proposition 6}
This proposition follows similar ideas as in the proof of Proposition 5. Some intermediate calculus used in the proof of Proposition 6 is showed below (we use the same algorithm for $k(n,e)$ and $p(n,e)$ as used for Proposition 5).

\includegraphics[width=1.0\textwidth]{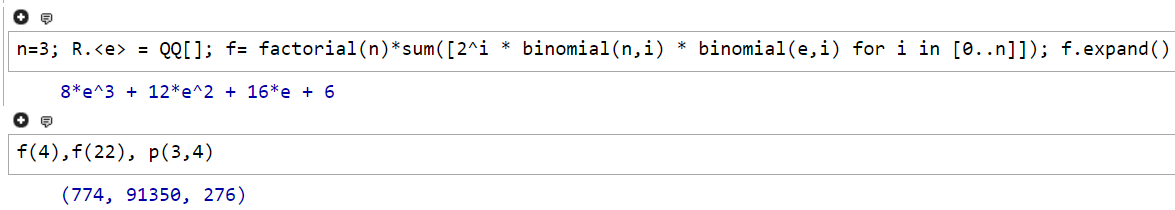}
\vspace{2mm}

\paragraph{Computation for Proposition 7}
Here we use the same code to compute $f(e)=12! \cdot k(12,e)=  12! \sum_{i=0}^{12} 2^{i} \binom{12}{i} \binom{e}{i}$ as in Proposition 5. We define $A=\{a': 9 \leq a' <18 \textrm{ or } 63\leq a'<72\}$ and have to check that $f(a'+3^4 b')\equiv 3^6$ or $2\cdot 3^6 \pmod{3^7}$ for every $a'\in A$ and $0\leq b'<3^3$. We check that it is true for each value of $a'+3^4 b'$ (there are $18\cdot 3^3= 486$ possibilities). 

\includegraphics[width=1.0\textwidth]{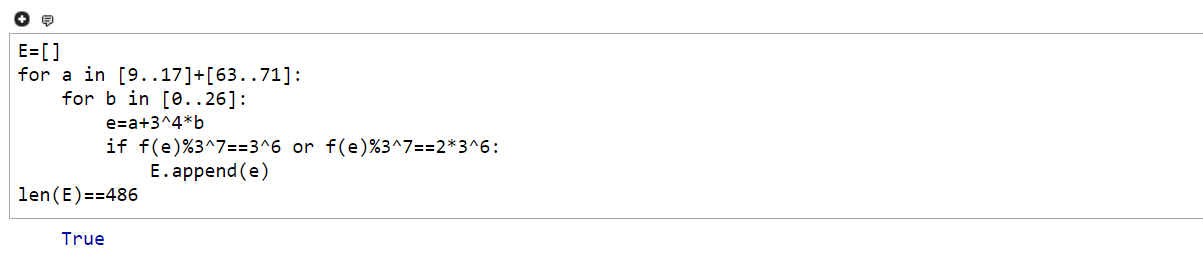}

We use the same code for $k(n,e)$ and $p(n,e)$ as in Proposition 5 and check that $p(12,a)\equiv 0 \pmod{3}$ for $9\leq a < 18$.

\includegraphics[width=1.0\textwidth]{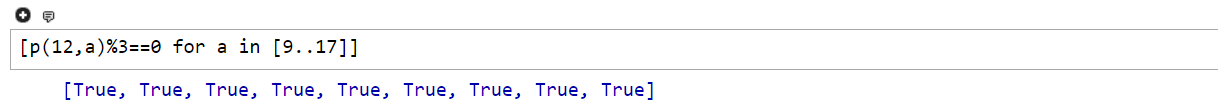}
\vspace{2mm}

\paragraph{Computation for Lemma 7}
We define the Davis-Webb symbols and $\eta(i_{m+1},i_m, i_{m-1}) = \simb{1,1}{i_{m+1},i_m}\cdot \simb{1}{i_m}^{-1}\simb{1,0}{i_m, i_{m-1}}$. Then, we compute $\eta(i_{m+1},i_m, i_{m-1})$ for $(i_{m+1},i_m, i_{m-1})=(0,0,0)$, $(0,0,1)$, $(0,0,2)$, $(0,1,0)$, $(0,1,1)$, $(0,1,2)$, $(0,2,0)$, $(0,2,1)$, $(0,2,2)$, $(1,0,0)$, $(1,0,1)$, $(1,0,2)$, $(1,1,0)$. We rename the variables in the SAGE code: $a=i_{m+1}, b=i_m$ and $c=i_{m-1}$.

\includegraphics[width=1.0\textwidth]{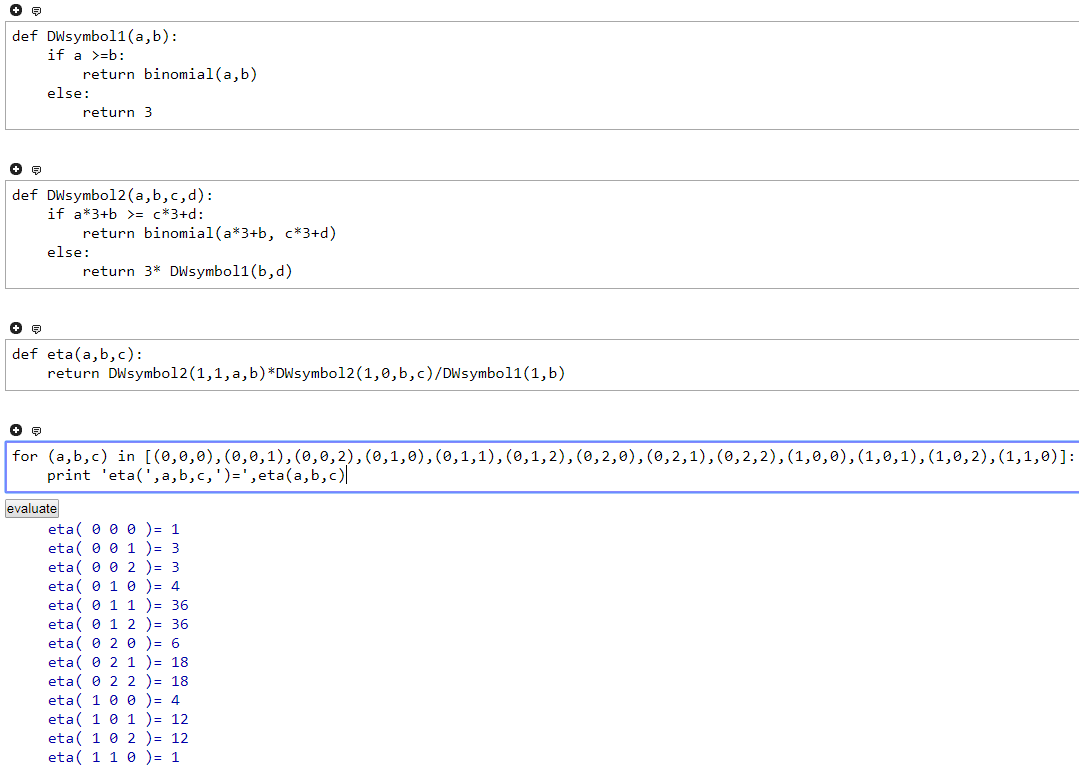}

\paragraph{Computation for Proposition 8}
We have $n=3^{m+1}+3^{m}$, $m\geq 2$ and $e=\sum_{i=0}^{h-1}e_i  3^i$ with  $h\geq m+4$ and $e_i\in \{0,2\}$ for $m+1<i\leq h-1$, $e_{m+1}=1$ and $e_i \in \{0,1,2\}$ for $0\leq i <m+1$. We define $\tilde{e} = e_{m+2}\cdot 3^4 + 3^3 + e_m\cdot 3^2 + e_{m-1}\cdot 3 + e_{m-2}$ and proved that $k(n,e)\equiv F(\tilde{e})\pmod{9}$ where $F(\tilde{e})=1 - 3 \binom{\tilde{e}}{3} + 3 \binom{\tilde{e}}{6} -4 \binom{\tilde{e}}{9} + 6 \binom{\tilde{e}}{18} - 4 \binom{\tilde{e}}{27} + 3 \binom{\tilde{e}}{30} - 3 \binom{\tilde{e}}{33} + \binom{\tilde{e}}{36}$. At the end of the proof of Proposition 8 we check that $F(\tilde{e}) \equiv \left\{  \begin{array}{ll}
3\!\!\! \pmod{9} & \textrm{if $27\leq \tilde{e}\leq 35$ or $207\leq \tilde{e}\leq 215$};\\
6\!\!\! \pmod{9} & \textrm{if $36\leq \tilde{e}\leq 53$ or $189\leq \tilde{e}\leq 206$}. \end{array}   \right.$
For this purpose we use the following SAGE code:

\includegraphics[width=1.0\textwidth]{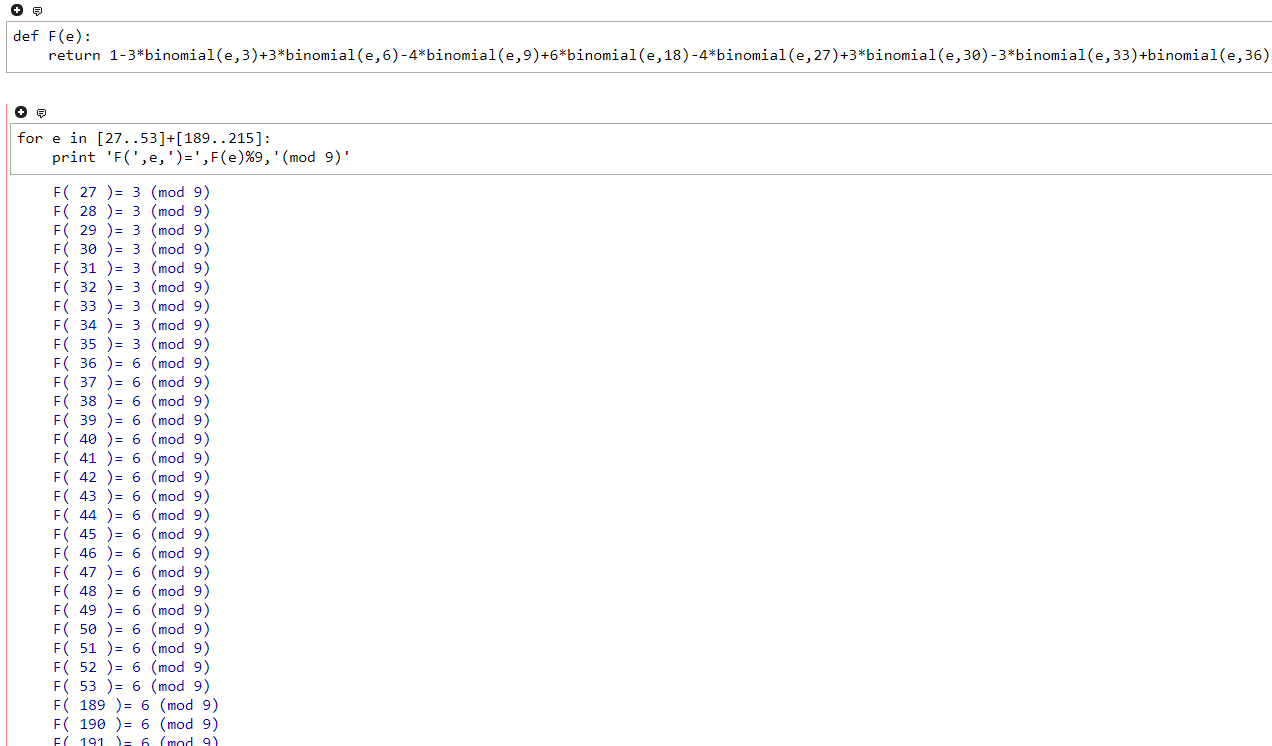}
\vspace{2mm}

\subsection*{Intermediate computation in Section IV}

\paragraph{Computation for Example 5}
We consider the Euclidean ball $B= \{b\in \Z^3: \sqrt{b_1^2+b_2^2+b_3^2}\leq 8 \}$ and calculate its cardinality and $p_1(B)=\sum_{b\in B}b_1^2$ by the following routine:

\includegraphics[width=1.0\textwidth]{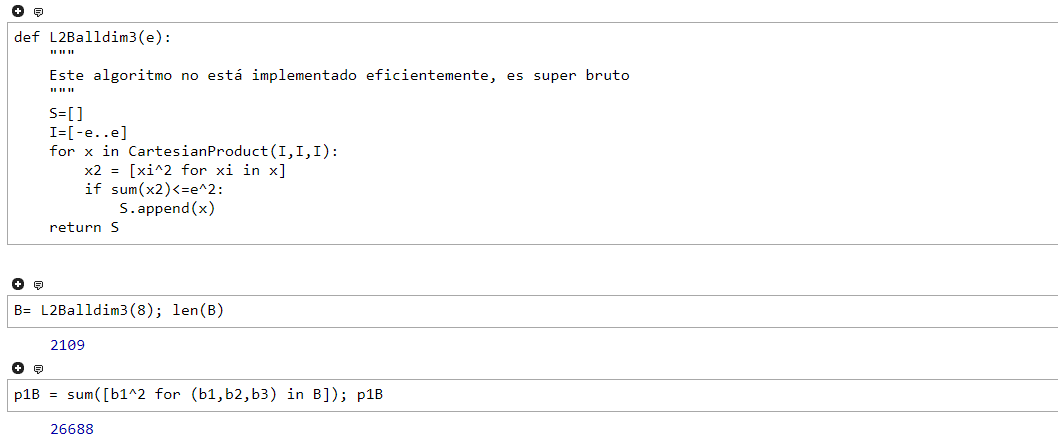}
\vspace{2mm}

\paragraph{Computation for Proposition 16}
First we calculate $k(n,6), p_1(n,6)$ and $p_2(n,6)$ using the following routine:

\includegraphics[width=1.0\textwidth]{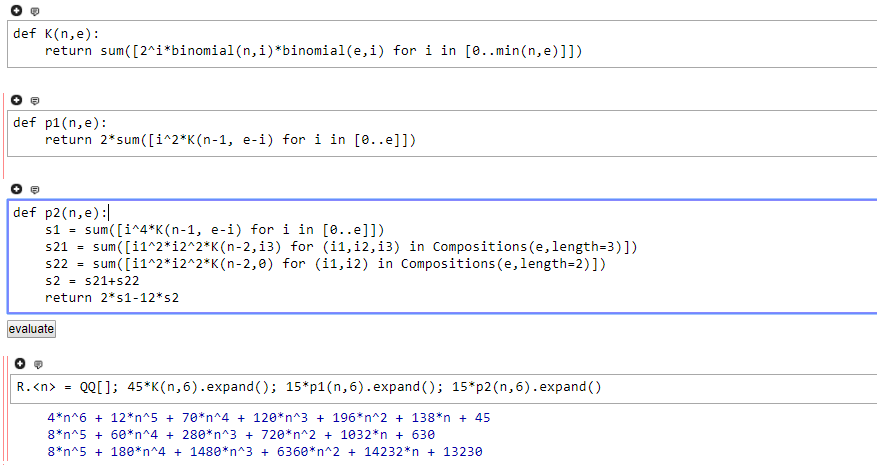}

Then we check the values of $n: 0\leq n < 125$ verifying the system $$\left\{ \begin{array}{l}
45k(n,6) \equiv 25,50,75\textrm{ or }100 \pmod{125};\\
15p_1(n,6) \not\equiv 0 \pmod{25};\\
15p_2(n,6) \equiv 0 \pmod{25}.
\end{array}  \right.  $$

\includegraphics[width=1.0\textwidth]{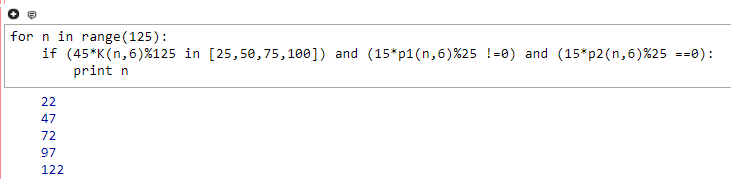}
\vspace{2mm}

\paragraph{Computation for Proposition 17}
First we calculate $k(n,7), p_1(n,7)$ and $p_2(n,7)$.

\includegraphics[width=1.0\textwidth]{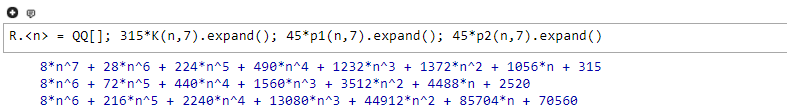}

Then we check the values of $n: 0\leq n < 125$ verifying the system $$\left\{ \begin{array}{l}
315k(n,7) \equiv 25,50,75\textrm{ or }100 \pmod{125};\\
45p_1(n,6) \not\equiv 0 \pmod{25};\\
45p_2(n,6) \equiv 0 \pmod{25}.
\end{array}  \right.  $$

\includegraphics[width=1.0\textwidth]{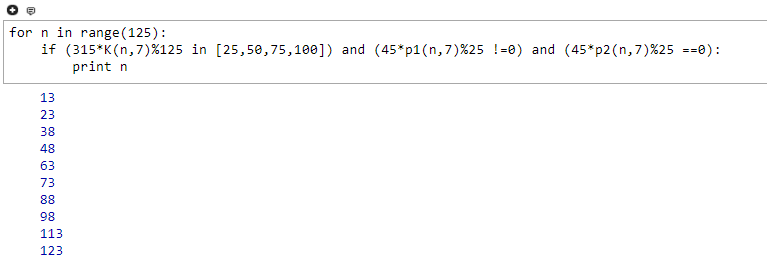}
\vspace{2mm}

\end{document}